\documentclass[11pt]{article}

\usepackage{float}

\usepackage{amsmath, amssymb,amsfonts,dsfont,xspace,graphicx,relsize,bm,mathtools,xcolor,amsthm,soul}
\usepackage{mathrsfs}
\usepackage[pagebackref]{hyperref}
\usepackage{enumerate}
\usepackage{lipsum} 

\def\01{\{0,1\}}

\newcommand{\eps}{\varepsilon}

\newcommand{\Tr}{\mbox{\sf Tr}}
\usepackage{tikz,tikz-qtree}
\usetikzlibrary{shapes.misc,shapes.geometric,arrows,positioning,calc,backgrounds}
\usepackage{wrapfig}

\usepackage{algorithm}% http://ctan.org/pkg/algorithms
\usepackage{algpseudocode}% http://ctan.org/pkg/algorithmicx

\newcommand{\E}{\mathcal{E}}
\newcommand{\Cc}{{\mathcal C}} %concept class, \C is already in use for complex nrs
\newcommand{\calA}{\ensuremath{\mathcal{A}}}
\newcommand{\calB}{\ensuremath{\mathcal{B}}}
\newcommand{\G}{\ensuremath{\mathcal{G}}}

\newcommand{\calC}{\ensuremath{\mathcal{C}}}
\newcommand{\In}{\ensuremath{\mathscr{I}}}

\newcommand{\DP}{\ensuremath{\mathsf{DP}}}
\newcommand{\PAC}{\ensuremath{\mathsf{PAC}}}
\newcommand{\PPAC}{\ensuremath{\mathsf{PPAC}}}

\newcommand{\ext}{\ensuremath{\mathsf{ext}}}
\newcommand{\SB}{\ensuremath{\mathsf{SB}}}
\newcommand{\PRD}{\ensuremath{\mathsf{PRdim}}}
\newcommand{\Loss}{\ensuremath{\mathsf{Loss}}}

\newcommand{\DRD}{\ensuremath{\mathsf{DRdim}}}

\newcommand{\RSOA}{\ensuremath{\mathsf{RSOA}}}

\newcommand{\GL}{\ensuremath{\mathsf{GL}}}
\newcommand{\sfat}{\ensuremath{\mathsf{sfat}}}
\newcommand{\Ldim}{\ensuremath{\mathsf{Ldim}}}
\newcommand{\srac}{\ensuremath{\mathsf{SRAC}}}
\newcommand{\rac}{\ensuremath{\mathsf{RAC}}}
\newcommand{\fat}{\ensuremath{\mathsf{fat}}}

\newcommand{\UB}{\ensuremath{\mathsf{UB}}}
\newcommand{\Hi}{\ensuremath{\mathcal{H}}}
\newcommand{\U}{\ensuremath{\mathcal{U}}}

\newcommand{\A}{\ensuremath{\mathcal{A}}}
\newcommand{\B}{\ensuremath{\mathcal{B}}}
\newcommand{\calR}{\ensuremath{\mathcal{R}}}
\newcommand{\T}{\ensuremath{\mathcal{T}}}

\newcommand{\M}{\ensuremath{\mathcal{M}}}

\newcommand{\X}{\ensuremath{\mathcal{X}}}
\newcommand{\Z}{\ensuremath{\mathds{Z}}}

\newcommand{\De}{\ensuremath{\mathcal{D}}}

\newcommand{\id}{\ensuremath{\mathbb{I}}}

\DeclareMathOperator{\poly}{poly}

\newcommand{\Fe}{\ensuremath{\mathcal{F}}}

\newtheorem{theorem}{Theorem}[section] 

\newtheorem{lemma}[theorem]{Lemma}
\newtheorem{claim}[theorem]{Claim}

\newtheorem{corollary}[theorem]{Corollary}
\newtheorem{result}[theorem]{Result}
\newtheorem{definition}[theorem]{Definition}
\newtheorem{problem}[theorem]{Problem}

\usepackage{theoremref}

\DeclareMathOperator*{\argmax}{arg\,max}
\DeclareMathOperator*{\argmin}{arg\,min}
\def\01{\{0,1\}}

\newcommand{\supp}{\mathsf{supp}}

\usepackage{wrapfig}

    \usepackage{authblk}
\usepackage{tikz,tikz-qtree}

\usepackage[margin=1in]{geometry}
\hypersetup{
	colorlinks,
	linkcolor={red!100!black},
	citecolor={blue!100!black},
}

\title{Private learning implies quantum stability}

\begin{document}

\author[1]{Srinivasan Arunachalam\footnote{\texttt{srinivasan.arunachalam@ibm.com}}}
%\email{Srinivasan.Arunachalam@ibm.com}
\affil{IBM Quantum, IBM T.J.\ Watson Research Center, Yorktown Heights, USA}
\author{Yihui Quek\footnote{\texttt{yquek@stanford.edu}}}
% \email{}
\affil{ Information Systems Laboratory, Stanford University, USA}
\author[1]{John Smolin\footnote{\texttt{smolin@us.ibm.com}}}

\date{\today}

\maketitle
\begin{abstract}
Learning an unknown $n$-qubit quantum state $\rho$ is a fundamental challenge in quantum computing. Information-theoretically, it is well-known that tomography requires exponential in $n$ many copies of an unknown  state $\rho$ in order to estimate it up to small trace distance.  
Motivated by computational learning theory, Aaronson and others introduced several (weaker) learning models: the $\PAC$ model of  learning quantum states (Proc.~of Royal Society A'07), shadow tomography (STOC'18) for learning ``shadows" of a quantum state, a learning model that additionally requires learners to be differentially private~(STOC'19), and  the online model of learning quantum states (NeurIPS'18). In  these models it was shown that an unknown quantum state can be learned ``approximately  well"  using \emph{linear} in $n$ many copies of~$\rho$. But is there any relationship between these learning models? In this paper we prove a sequence of (information-theoretic)  implications from differentially-private $\PAC$ learning to online learning and then to quantum stability.

\vspace{2mm}

Our main result generalizes the recent work of Bun, Livni and Moran~(Journal of the ACM, 2021) who showed that finite Littlestone dimension (of Boolean-valued concept classes) implies $\PAC$ learnability in the (approximate) differentially private ($\DP$) setting.  We first consider their work in the real-valued setting and further extend to their techniques to the setting of learning quantum states. Key to many of our results is our construction of a generic quantum online learner, Robust Standard Optimal Algorithm ($\RSOA$), which is robust to adversarial imprecision. We then show information-theoretic implications between $\DP$ learning quantum states in the~$\PAC$ model, learnability of quantum states in the one-way communication model, online learning of quantum states, quantum stability (which is our new \emph{conceptual} contribution) and various combinatorial parameters.  As an application, we also improve gentle shadow tomography (for classes of quantum states) and show connections between noisy quantum state learning and channel capacity, which might be relevant to physically-motivated learning~scenarios.
\end{abstract}
\newpage

     \setcounter{tocdepth}{2}
    \renewcommand{\baselinestretch}{0.7}\normalsize
    \tableofcontents
    \renewcommand{\baselinestretch}{1.0}\normalsize

% \section{Introduction}
\section{Introduction}
Quantum state tomography is a fundamental task in quantum computing whose goal is to estimate an unknown quantum state $\rho$, given copies of the state.
 Tomography is of great practical interest since it helps in tasks such as verifying entanglement, understanding correlations in quantum states, and is useful for  calibrating, understanding and controlling noise in   quantum devices.
In the last few years,  questions about the fundamental limits of this task have gained a lot of theoretical attention, in particular, how many copies of an $n$-qubit quantum state~$\rho$ are necessary and sufficient to estimate the density matrix $\rho$ up to small error? In this direction, recent breakthrough results of~\cite{o2016efficient,DBLP:conf/stoc/ODonnellW17,haah2017sample} showed that ${\Theta}(2^{2n}/\varepsilon^2)$ copies of $\rho$ are \emph{necessary and sufficient} to learn~$\rho$ up to trace distance~$\varepsilon$. Unfortunately, this exponential scaling in complexity is reflected in practical applications of tomography;  the best known experimental implementation of full-state quantum tomography has been for a 10-qubit quantum state~\cite{song201710}. Moving beyond 10-qubits is also hard since the theoretical guarantees of tomography require millions of copies of the unknown state~\cite{o2016efficient,DBLP:conf/stoc/ODonnellW17,haah2017sample} in order to {\em fully} characterize it -- a formidable tax on resources.

This raises the natural question: is it always necessary for experimental purposes to estimate the {\em full} density matrix of $\rho$? Rather than learning $\rho$ up to trace distance $\varepsilon$, do there exist weaker but still practically useful learning goals, which would enable savings in sample complexity?  These questions have turned the attention to  `essential' models of learning, which aim to learn only the useful properties of a unknown quantum state. In this direction, a few works have introduced models for learning quantum states, inspired by classical computational learning theory. In this paper, we show (information-theoretic) implications between these seemingly different quantum learning models.

\subsection{Background: Models of interest}
 To explain our main results, we start by introducing some learning models of interest.  Below we describe  the PAC learning model, online learning model, learning under differential privacy constraints and one-way communication complexity for learning quantum~states. We formally define these models in the Section~\ref{subsec:learning_models}. 
\vspace{1mm}

\noindent\textbf{PAC learning.} Probably Approximately Correct ($\PAC$) learning, introduced by~\cite{valiant1984theory}, lays the foundation for computational learning theory.~\cite{aaronson2007learnability} considered the natural analog of learning quantum states in the $\PAC$ model. In this model, let $\rho\in \calC$ be an unknown quantum state (picked from a \emph{known} \emph{concept class}~$\calC$ of states) and let $D:\E\rightarrow [0,1]$ be an \emph{arbitrary} unknown distribution over all possible $2$-outcome measurements $E$. Suppose a quantum learner obtains \emph{training examples} $(E_i,\Tr(\rho E_i))$ where $E_i$ is drawn from $D$, and the goal is to output~$\sigma$ such that with probability at least $0.99$, $\sigma$ satisfies $\Pr_{E\sim D}[|\Tr(\sigma E)-\Tr(\rho E)|\leq \zeta]\geq 1-\alpha$ (this second probability is over a fresh example from $D$). How many training examples suffice for such a $(\zeta,\alpha)$-$\PAC$ learner? In answer,~\cite{aaronson2007learnability} showed that the number of examples necessary and sufficient to learn~$\calC$ is captured by the \emph{fat-shattering dimension} of $\calC$.

\vspace{1mm}

\noindent\textbf{PAC learning with Differential privacy.}  A well-studied area of computer science is differential privacy ($\DP$) (which says that an algorithm should behave ``approximately" the same given two datasets that differ in one element). This notion can be extended to the quantum realm, where we ask that the quantum $\PAC$ learner proposed above is also \emph{differentially private}, wherein given two datasets $S=\{(E_i,\Tr(\rho E_i))\}_i$,  $S'=\{(E'_i,\Tr(\rho E'_i))\}_i$ such that there exists a unique $i$ such that $E_i \neq E'_i$,\footnote{In Boolean $\DP$, one allows the examples to be the same and labels be different for neighboring databases. Since we look at  real-valued $\DP$, we consider the case when the examples are different.} then  a quantum  $(\gamma,\delta)$-$\DP$ $\PAC$ learning algorithm needs to satisfy
$
\Pr[\calA(S)=\sigma]\leq e^{\gamma} \Pr[\calA(S')=\sigma]+\delta,
$
where $\calA(S)$ is the output of~$\calA$ on input $S$.\footnote{Our notion of $\DP$ differs from the notion of $\DP$ proposed by~\cite{aaronson2019gentle}. They consider \emph{$\DP$ measurements} with respect to a class of \emph{product states}, whereas here we require $\DP$ with respect to the dataset $\{(E_i,\Tr(E_i\rho)\}_i$, which naturally quantizes the classical  definition of differential privacy.}
\vspace{1mm}

\noindent\textbf{Communication complexity.} Consider the standard one-way communication model between Alice and Bob. Suppose Alice has a quantum state $\rho$ (unknown to Bob) and Bob has an unknown (to Alice) measurement $E$. The goal of Bob is to output an approximation of $\Tr(\rho E)$ if only Alice is allowed to communicate to Bob. A trivial strategy for this communication task is for Alice to send a classical description of $\rho$, but can we do better? If so, how many bits of communication suffice for this task? 
\vspace{1mm}

\noindent\textbf{Online learning.} Several features of the $\PAC$ quantum learning model and tomography are somewhat artificial: first, the assumption that the measurements (training examples) are drawn from the same unknown distribution $D$ that the learner will be evaluated on, which does not account for adversarial or changing environments, and secondly, it may be infeasible to possess $T$-fold tensor copies of the unknown quantum state $\rho$, rather we may only be able to obtain sequential copies of it. The quantum online learning model addresses these aspects. Online learning consists of repeating the following rounds of interaction: for an unknown state $\rho$, at every round a  maintains a local $\sigma$ which is its guess of $\rho$, obtains a description of measurement operator $E_i$ (possibly adversarially) and predicts the value of $y_i=\Tr(\rho E_i)$. Subsequently it  receives as feedback an $\varepsilon$-approximation of $y_i$. On every round, if the learner's prediction satisfies $|\Tr(\sigma E_i)- y_i|\leq \varepsilon$ then it is \emph{correct}, otherwise it has made a \emph{mistake}. The goal of the learner is the following: minimize $m$ so that after making $m$ mistakes (not necessarily consecutively), it makes a correct prediction (i.e., approximates $\Tr(\cdot \rho)$ well-enough) on \emph{all} future rounds. 

Importantly, while the goal in the above model is to make real-valued predictions $y_i$, it departs from the real-valued online learning literature in allowing for $\eps$-imprecision in the feedback. This imprecision is inherent to all learning settings where the feedback is generated by a statistical algorithm or physical measurements (in the quantum learning setting, the feedback arises from processing the outcomes of quantum measurements), and this generalization has non-trivial implications, as we show. Working in this model,~\cite{aaronson2018online} showed that for learning the class of all quantum states, it suffices to let $m$ be at most  \emph{sequential fat-shattering dimension} of $\Cc$ (a combinatorial parameter which was originally introduced in the classical work of~\cite{rakhlin2010online}).
\vspace{1mm}

\noindent All these  learning models can be seen as variants of full-state tomography, and are known to require \emph{exponentially} fewer resources than tomography. A natural question is:
\begin{center}
    \emph{Is there a relation between these {learning} models, {communication} and {combinatorial}~parameters?}
\end{center}

Understanding this question classically in the context of Boolean functions has received tremendous attention in computational learning theory and theoretical computer science in the last two years. There have been a flurry of papers establishing various connections~\cite{bun2020equivalence,jung2020equivalence,ghazi:sampleproperPAC,alon2019private,bousquet2019passing,bun2020computational,gonen2019private,alon2020closure,haghtalab2020smoothed}. However  understanding if the results in these papers apply to the quantum framework has remained unexplored.

\subsection{Overview of main results}

To condense our (affirmative) answer to the question above, we derive a series of implications going through all these models, starting from differentially-private $\PAC$ learning to online learning to quantum stability (our conceptual contribution which we define and discuss below).

  {\renewcommand{\arraystretch}{1.4} 
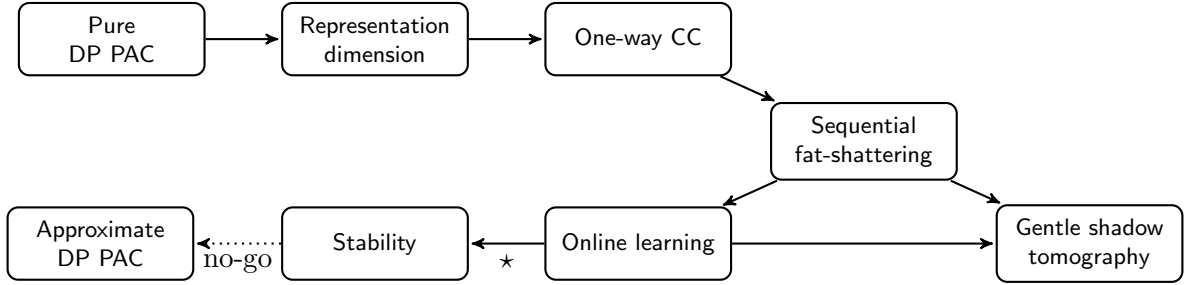
\begin{figure*}[!ht]
\centering
\begin{tikzpicture}
[->,>=stealth',shorten >=1pt,auto,  thick,yscale=0.8,
main node/.style={circle,draw}, node distance = 1cm and 1.8cm,
block/.style   ={rectangle, draw, text width=5.5em, text centered, rounded corners, minimum height=2.5em, fill=white, align=center, font={\footnotesize}, inner sep=5pt}]
        \node[main node,block,] (PureDPPAC) at (-4,2.2) {\textsf{Pure\\ DP PAC}};
    \node[main node,block] (Prdim) at (-0.5,2.2) {\textsf{Representation dimension}};
    \node[main node,block] (Roneway) at (3,2.2) {\textsf{One-way CC}};
   \node[main node,block] (sfat) at (6,0.5) {\textsf{Sequential fat-shattering}};
  \node[main node,block] (online) at (3,-1.2) {\textsf{Online learning}};
  \node[main node,block] (stability) at (-0.5,-1.2) {\textsf{Stability}};
  \node[main node,block] (Appdim) at (-4.15,-1.2) {\textsf{Approximate DP PAC}};
    \node[main node,block] (shadow) at (9,-1.2) {\textsf{Gentle shadow tomography}};
    \path [->] (PureDPPAC) edge node {} (Prdim);
    \path [->](Prdim) edge node {} (Roneway);
    \path [->](Roneway) edge node {} (sfat);
    \path [->](sfat) edge node {} (online);
    \path [->](sfat) edge node {} (shadow);
    \path [->](online) edge node {} (shadow);
    \path [->](online) edge node {$\star$} (stability);
   \path [dotted,->](stability) edge node {no-go} (Appdim);
\end{tikzpicture}
    \caption{High-level summary of results relating models of learning. These results apply to the setting of learning real-valued classes and quantum states with imprecise feedback. Except for the $\star$-arrow, an arrow $\textsf{A}\rightarrow \textsf{B}$ in the figure implies that, if the sample complexity of learning in model $\textsf{A}$ or the combinatorial parameter $\textsf{A}$ is~$S_{\textsf{A}}$, then the complexity of learning in model $\textsf{B}$ or the combinatorial parameter $\textsf{B}$ is $S_{\textsf{B}}=\poly(S_{\textsf{A}})$. For the $\star$-arrow, the overhead is $S_{\textsf{B}}=\exp(S_{\textsf{A}})$. The dotted arrow signifies that a technique used to prove that corresponding implication for Boolean function classes is a no-go for our quantum learning setting.
    }
    \label{fig:Dependencies}
\end{figure*}}

Taking a step back, this is surprising:  quantum online learning and $\DP$ $\PAC$ quantum learning seem very different on the surface. Online learning ensures that {\em eventually}, after a certain number of  mistakes, we have learned the state up to trace distance. $\DP$ $\PAC$ learning is {\em not} online -- it separates the learning into train (offline) and test (online) phases, and also introduces a distribution, $D$, from which measurements are drawn.  Ultimately $\DP$ $\PAC$ learning says that after seeing $T$ measurements from $D$, we have (privately) learned the state.\footnote{At least, well-enough to predict its behavior on future measurements from $D$ with high probability.} We show that in fact $\DP$ $\PAC$ learning's sample complexity can be lower bounded by the {\em sequential fat-shattering dimension}
which also characterizes the complexity of online learning \cite{rakhlin2010online}. We give a high-level summary of our results in Figure~\ref{fig:Dependencies} (we say an algorithm is \emph{pure} $\DP$ (resp.~\emph{approximate} $\DP$) when $\delta=0$ in our definition (resp.~$\delta >0$). We remark that only a few of the arrows are efficient in both sample and time complexity, otherwise these implications are primarily information-theoretic. 

\noindent While some of these implications are known classically (for Boolean functions), our quantum learning is concerned with learning real-valued functions with \emph{imprecise} feedback. This difference has non-trivial consequences (which we highlight later), one of which is that a technique due to~\cite{bun2020equivalence} showing stability implies approximate $\DP$ $\PAC$ for Boolean functions, is a no-go in our~setting, as we show later.

\paragraph{Conceptual contribution.} The main center piece in establishing these connections is the concept of \emph{quantum stability}, which is the new conceptual contribution in this work. Intuitively, we say a quantum learning algorithm is stable if,  for an unknown state $\rho$, given a set of noisy labelled examples drawn i.i.d.~from a distribution $D$, there exists one state $\sigma$ such that, with ``high" probability, the output of the learning algorithm is ``close" to $\sigma$. More formally, we say a learning algorithm~$\calA$ is $(T,\varepsilon,\eta)$-\emph{stable} with respect to distribution $D$ if, given $T$ many labelled examples $S$ consisting of $E_i$ and approximations of $\Tr(\rho E_i)$, there exists a state $\sigma$ such~that 
\begin{equation}\label{eq:stable}
\Pr[\calA(S) \in \calB(\varepsilon, \sigma)]\geq \eta,
\end{equation}
where the probability is taken over the examples in $S$ and $\calB(\varepsilon,\sigma)$ is the ball of states within trace distance $\varepsilon$ of $\sigma$. In other words, quantum stability means that up to an $\eps$-distance, there is some $\sigma$ that is output by $\calA$ with probability at least $\eta$. 

While we will make this precise later, the significance of an algorithm $\calA$ being {\em stable} is that~$\sigma$, the output state at the `center of the ball', is a good hypothesis for estimating measurement probabilities (and hence $\calA$ is a good learner). This is not at all obvious from the definition of stability, which does not inherently require that this $\sigma$ is a good approximation of $\rho$. Yet, we show that if $\mathcal{A}$ is a stable and consistent learner (i.e., its output does not contradict any of the training examples it has seen), $\sigma$ has low loss with respect to $D$. This means that, using hypothesis $\sigma$ to predict outcomes of future measurements drawn from distribution $D$ as $\Tr(E\sigma)$ will yield $\eps$-accurate predictions with high~probability. 

Classically, stability is conceptually linked to differential privacy. In fact,~\cite{dwork2014} state that ``\emph{Differential privacy is enabled by stability and ensures stability...we observe a tantalizing moral equivalence between learnability, differential privacy, and stability,}" and this notion was crucially~used in~\cite{bun2020equivalence,alon2019private,abernethy2019online,bousquet2019passing}.  Such a connection has remained unexplored (and even undefined) in the quantum setting and in this work we explores this interplay between privacy, stability and quantum~learning. Our definition of stability marks a crucial departure from the classical notion of stability used by~\cite{bun2020equivalence}, in the following sense: a classical learning algorithm is stable if a \emph{single} function is output by the algorithm with high probability. In contrast we say that a quantum learning algorithm is stable if a \emph{collection} of quantum states is output with high probability. Given the significance of the notion of stability in classical $\DP$ research, we believe our definition could find further applications in quantum computing.

\subsection{Proof techniques and further contributions}
We break down the proof of the arrows in Figure~\ref{fig:Dependencies} into four steps and discuss them below. 

\paragraph{1. Pure $\DP$ $\PAC$ implies finite sequential fat-shattering dimension.}  It is well-known classically that if there is a $\DP$ $\PAC$ learning algorithm for a class $\calC$ then the \emph{representation dimension} of the class is small (we define this dimension formally in Definitions~\ref{def:DRD},~\ref{def:PRD}). Representation dimension is then known to upper bound classical communication complexity as well as a combinatorial dimension of the concept class known as the {\em sequential fat-shattering dimension} $\sfat(\calC)$. All of the above connections are classical, but we show that they can be ported to learning classes of quantum states. 
To do so, we make all implications mentioned above robust to our `quantum' version of $\DP$  $\PAC$ learning, that is for real-valued functions and with adversarial imprecision. One particular contribution in this direction is: prior to our work, Feldman and Xiao~\cite{feldman2014sample} showed that representation dimension is only a lower bound for one-way \emph{classical} communication complexity, but we show that this dimension even lower bounds \emph{quantum} communication complexity. For the case of \emph{Boolean} valued concept classes,  Zhang~\cite{Zhang2011OnTP} proved a weaker version of our result relating Littlestone dimension and communication complexity, and the  proof of our main result easily recovers his result, significantly simplifies and extends his proof. 

For the remaining part of the introduction, we slightly abuse notation: for a class of quantum states $\Cc$, we define $\sfat(\Cc)$ as the sequential fat shattering dimension, not of the class $\Cc$, but of an associated class of functions acting on the domain $\M$ of all possible $2$-outcome measurements on states in $\Cc$. To be precise, for every $\Cc$, we associate the real-valued concept class $\Fe_\Cc=\{f_\rho: \M\rightarrow [0,1]\}_{\rho \in \Cc}$ where $f_\rho(E)=\Tr(E\rho)$ for all $E\in \M$. The $\zeta$-\emph{sequential fat shattering dimension} (denoted $\sfat_\zeta(\cdot)$) of the class of states $\Cc$  is defined in terms of the class of real-valued functions $\Fe_\Cc$, i.e., $\sfat_\zeta(\Cc):=\sfat_\zeta(\Fe_\Cc)$ (we define this combinatorial parameter more formally in Section~\ref{sec:prelim}).

\paragraph{2. Finite $\sfat(\cdot)$ implies online learning.} 
In the second step, the goal is to go from a concept class $\calC$ having finite $\sfat(\calC)$ to design an online learning algorithm for $\calC$ that makes at most $\sfat(\calC)$ mistakes.  In this direction, one of our technical contributions is to construct a \emph{robust} standard optimal algorithm (denoted $\RSOA$) which satisfies this mistake-bound. This $\RSOA$ algorithm summarized in the result below will be crucial for the following steps. 
\begin{result}[Informal]
\label{thm:rsoaresult}
Let $\Cc$ be a  class of quantum states with $\sfat_\zeta(\Cc)=d$.  There is an explicit robust standard optimal algorithm $\RSOA$ that makes at most $d$ mistakes in online learning~$\Cc$. 
\end{result}
We now make a few comments regarding this result. Classically, for the Boolean setting, it is well-known that the so-called Standard Optimal Algorithm is an online learner for \emph{any} concept class $\Cc$, that makes at most Littlestone dimension of $\Cc$-many mistakes~\cite{littlestone}. Eventually, Rakhlin et al.~\cite{rakhlin2010online} generalized the work of Littlestone for real-valued functions, showing that real-valued concept classes can be learned using their $\textsf{FAT-SOA}$ algorithm, with at most $\sfat(\Cc)$ many~mistakes. Now, our $\RSOA$ algorithm generalizes this, showing that real-valued concept classes can be learned with $\sfat_{\zeta}(\Cc)$ many mistakes, even in the presence of adversarial imprecision of magnitude $\zeta$ (which is the case for quantum learning).

The following basic principle underlies our $\RSOA$ algorithm and previous algorithms: after every round during which the learner obtains an $x$ and $\zeta$-approximation of $c(x)$, eliminate the concepts in~$\Cc$ that are ``inconsistent" with the adversary's feedback. In more detail: first, the learner discretizes the function range $[0,1]$ into $1/\zeta$-many $\zeta$-sized bins. In the first round of learning, upon receiving the domain point $x$, the learner evaluates all the functions in $\Cc$ at $x$ and `counts' (using $\sfat(\cdot)$ dimension as a proxy) the number of functions mapping to each bin, and chooses the bin (and outputs the midpoint of this bin as its guess for $c(x)$)  with the highest $\sfat(\cdot)$ dimension.  After the learner obtains a $\zeta$-approximation of $c(x)$, it removes those functions in $\Cc$ that were inconsistent with this approximation and then proceeds to the next round. After $\sfat(\Cc)$ many mistakes, one can show that the learner identifies the unknown~concept.  We remark that in prior works, it was assume that the learner obtained $c(x)$ exactly, whereas here it only receives $\zeta$-approximations. This robustness property allows us to use $\RSOA$ in the context of learning quantum states, where typically, the feedback is generated by measuring $E$ repeatedly on copies of the quantum state $\rho$, which will provides a $\zeta$-approximation of $\Tr(\rho E)$.

Prior to our work, Aaronson et  al.~\cite{aaronson2018online} showed that $\sfat(\cdot)$ of all $n$-qubit quantum states is $n$, which implies the \emph{existence} of a quantum online learning algorithm for the class of all quantum states that makes at most $n$ mistakes. However, their focus was on quantum online learning with {\em regret} bounds, and so they never provided an explicit algorithm that achieves the $\sfat(\cdot)$ mistake bound, and raised this as an open question. Our Result~\ref{thm:rsoaresult} resolves their question, by showing that our $\RSOA$ algorithm can online learn a class of quantum states $\Cc$ by making $\sfat(\Cc)$ many~mistakes.

\paragraph{3. Online learning implies stability.}

We now show that if a concept class  $\Cc$ satisfies $\sfat(\Cc)=d$ (i.e.,      it can be online-learned with $d$ many mistakes), then $\Cc$  can be learned by a globally-stable quantum algorithm with parameters $(\varepsilon^{-d},\varepsilon,\varepsilon^d)$, i.e.,      there exists a stable learner that, given $\varepsilon^{-d}$ many examples $(E_i,\Tr(\rho E_i))$,  with probability at least $\varepsilon^d$, outputs a state $\sigma$ that is $\varepsilon$ close to the unknown target state $\rho$.

\begin{result}
\label{res:global_stability}
Let $\Cc$ be a class of quantum states with $\sfat_{\zeta}(\Cc) = d$. There exists an algorithm~$\G$ that satisfies the following: for every $\rho\in \Cc$, given $T=\zeta^{-d}/\varepsilon$ many labelled examples $S$ consisting of $E_i$ drawn from a distribution $D$ over a set of orthonormal\footnote{For simplicity, we have required that the measurements are drawn from an orthonormal set. This is only necessary if we require the algorithm $G$ to be a proper learner, that is, its output function $f$ is guaranteed to be such that one can always construct a density matrix $\sigma$  for which $\Tr(\sigma M)=f(M)$ for all $M\in \M$.} $2$-outcome measurements and $\zeta$-approximations of $\Tr(\rho E_i)$, there exists a $\sigma$ such that $\Pr_{S\sim D^T}[\G(S) \in \B_{\M}(\zeta,\sigma)] \geq \zeta^d$ and $\Pr_{E\sim D}\big[|\Tr(\rho E)-\Tr(\sigma E)|\leq \zeta\big] \geq 1-\varepsilon$.
\end{result}

In order to prove this theorem we borrow the high-level idea from~\cite{bun2020equivalence} (for the case of Boolean functions). Like \cite{jung2020equivalence} which studied the case of online multi-class regression, we borrow the high-level idea from Bun et al.~\cite{bun2020equivalence} (originally developed for the case of Boolean functions) to construct our stable learner: we sample many labelled examples from the distribution $D$ and instead of feeding these examples directly to the black-box $\RSOA$, we plant amongst them some ``mistake examples" before giving the processed sample to $\RSOA$. A ``mistake example" is an example which is correctly labelled, but on which $\RSOA$ would make the wrong prediction. That is to say, from a large pool of $T = \zeta^{-d}$ examples drawn from $D$, craft a short sequence of $O(1/\zeta)$ examples that include at most $d$ mistake examples; now feed the short sequence into $\RSOA$. This works, because $\RSOA$ satisfies the guarantee (Result~\ref{thm:rsoaresult}) that after making $d = \sfat_\zeta(\Cc)$ mistakes, it would have completely identified the target concept.

We proceed similarly but tackle some subtleties related to our learning setting. First, before our work we didn't have an $\RSOA$ algorithm which could be used as a black-box in order to emulate the proof-technique of~\cite{bun2020equivalence} for the case of quantum states. Second, our technique for creating ``mistake examples" differs from that of~\cite{bun2020equivalence}. In the Boolean case, to insert a mistake, it suffices to do the following: suppose $c:\mathcal{X} \rightarrow \{0,1\}$ is the unknown target function, they take two candidate set of examples $S_1, S_2$, run two parallel runs of $\RSOA$ on $S_1,S_2$, and obtain two output hypothesis functions $f_1, f_2$. They then identify a point in the domain $x$ at which $f_1(x) \neq f_2(x)$; since they are Boolean functions, one must evaluate to $c(x)$ and the other to $\overline{c(x)}$. Say $f_1(x) = c(x)$ and $f_2(x) = \overline{c(x)}$, i.e., the hypothesis $f_2$ makes a \emph{mistake} at $x$. They append a ``mistake" example of the form $(x,c(x))$ to $S_2$, so that when $\RSOA$ is now run on $S_2 \circ (x,c(x))$, $\RSOA$ is forced to make a \emph{new} mistake on this new set of examples. A subtlety here is, a learner does not know $c(x)$ and~\cite{bun2020equivalence} simply flip a coin $b\in \01$ and let the mistake example be  $(x,b)$ so that with probability $1/2$, $b = c(x)$. For us this does not work because our target function is real-valued, i.e., $c(x)\in [0,1]$. Instead, we discretize  $[0,1]$ into $1/\zeta$ many $\zeta$-intervals, pick a uniformly random interval and let $b$ be the center of this interval.  Clearly now, with probability $1/\zeta$, $c(x)$ lies in the $\zeta$-ball around $b$, but still is not  \emph{equal} to $b$. The construction of our quantum stable learner and the analysis is more  involved to overcome this issue and errors in the adversary feedback.

\paragraph{4. Stability  does not imply approximate $\DP$ $\PAC$ (without  a domain size dependence).} 
So far we showed that quantum online learning $\Cc$ implies the existence of a globally stable learner (with appropriate parameters) for $\Cc$.  For Boolean-valued $\Cc$s, in~\cite{bun2020equivalence}  they went one step further and created a \emph{approximately differentially-private} learner from a stable learner; in this sense, stability can be viewed as an intermediate property between online learnability and differential privacy. A natural question here is, can we extend this result to our setting, i.e., we showed earlier quantum online learning implies stability, but does quantum stability imply quantum differential privacy? If such a result also worked for the quantum setting (or real-valued functions), then Figure~\ref{fig:Dependencies} would start and end with differential privacy (albeit starting from pure $\DP$ and resulting in approximate $\DP$) and answer the question ``what can be privately quantum-learned" (akin to the classical work of~\cite{kasiviswanathan2011can}).

In this work, we show that, one cannot go from a stable learner to a differentially private learner without a domain-size dependence. 
First observe that our ``stability" guarantees on $\G$ (Result~\ref{res:global_stability}) are somewhat unusual: there exists some {\em function ball} (around the target concept) such that the collective probability of $\G$ outputting its member functions is high, in contrast to the Boolean setting~\cite{bun2020equivalence}, where global stability means that a {\em single} function is output with high probability. Again, this difference in our definition of global stability is because we only require that our real learner outputs a pointwise approximation of the target function $c$ -- namely one that is in the ball of $c$.    In the Boolean setting, to convert a stable learner to a private learner,~\cite{bun2020equivalence} used  stable histograms algorithm~\cite{bun2019simultaneous} and the generic private learner~\ref{lem:GL} and obtained a private learner with sample complexity depending on on $\Ldim(\Cc)$, the privacy, accuracy parameters of the stable learner, but \emph{not} the domain size of the function class. 

 Now, in our quantum setting since the learner only obtains an $\eps$-accurate feedback from the adversary, we  allow the learner to output a function in the $\varepsilon$-ball around the target concept. We also show that the generic transformation from a stable learner to a private learner doesn't work for the real-valued setting (in particular also quantum setting). The idea to show this is that this problem is a general case of the {\em one-way marginal release} problem, whose complexity needs to depend on the domain size. In particular, for learning quantum states on the Pauli observables on $n$ qubits, the sample complexity of the $\DP$ $\PAC$ learner will depend exponentially in $n$. The proof of this lower bound uses ideas from classical fingerprinting codes~\cite{bun2018fingerprinting} (which were also used earlier by Aaronson and Rothblum~\cite{aaronson2019gentle} in order to give lower bounds on gentle shadow~tomography).\footnote{This argument was communicated to us by Mark Bun~\cite{bununpublished}.}

\paragraph{Comparison to prior work~\cite{jung2020equivalence}.} After completion of this work, we were made aware by an anonymous referee of the paper by Jung, Kim and Tewari~\cite{jung2020equivalence} that extends the work of Bun et al.~\cite{bun2020equivalence} to multi-class functions (i.e., when the concept class to be learned $\Cc$ maps to a discrete set $\{1,\ldots,k\}$). They claim that their results also apply to real-valued learning by discretizing the range of the functions (we couldn’t find a
version of the paper that spells out the proof that online learnability implies a stable real-valued learner, but this seems implicit from their proofs). Despite this similarity, our quantum learning setting and resulting analysis differs from theirs in several crucial ways, which we now outline. 

Firstly,~\cite{jung2020equivalence}'s notion of stability for learning real-valued functions resembles our definition, \emph{however} in order to prove that online learnability implies stability, they crucially rely on a {\em modified} Littlestone dimension. In this work, we use the standard notion of $\sfat(\cdot)$ -- which we also bound in the case of quantum states -- and still show this implication. Secondly, for both $\PAC$ learning and online learning settings,~\cite{jung2020equivalence} assume that the feedback received by the learner is {\em exact}, i.e., for online learning, on input $x$, the adversary produces $c(x)\in [0,1]$ and for $\PAC$ learning, the examples are of the form $(x, c(x))$. By contrast, in this work, we only assume that the feedback in all learning models we consider (which includes both these settings) is a $\varepsilon$-approximation of $c(x)$. This generalizes the previous settings and arises from the fact that, in quantum learning, the feedback comes from some quantum estimation process or quantum measurement. Thus, all implications proven in this work are robust to such adversarial imprecision. This imprecision crucially bars the usage of \cite{bun2020equivalence}'s technique, developed for Boolean functions as a black-box, to conclude approximate $\DP$ $\PAC$ learning from stable learning.

\subsection{Further applications}
\label{sec:furtherapplications}
We now highlight a few applications of the consequences we established above. 

\textbf{\emph{1}. Faster shadow tomography for classes of states.}  Aaronson~\cite{aaronson:shadow} introduced a learning model called \emph{shadow tomography}. Here, the goal is to learn the ``shadows" of an unknown quantum state $\rho$, i.e.,  given $m$ measurements $E_1,\ldots,E_m$, how many copies of $\rho$ suffice to estimate $\Tr(\rho E_i)$ for all $i\in [m]$. Aaronson surprisingly showed that $O(n,\log m)$ copies of $\rho$ suffice for this task, and an important open problem was (and remains) can we get rid of the $n$ dependence in the complexity (even for a class of interesting quantum states)\footnote{We remark that Aaronson's model~\cite{aaronson:shadow} is not concerned with specific classes of quantum states, and instead considers learnability of an arbitrary quantum state. Nevertheless it is often reasonable to assume we have some prior information on the state to be learned, which means that it comes from a smaller class.}? Subsequently, in a recent work, Aaronson and Rothblum~\cite{aaronson2019gentle} also showed that learnability in the \emph{online} setting can be translated to algorithms for \emph{gentle} shadow tomography (in an almost black-box fashion). 
In this work, we use our results on quantum online learning and ideas in~\cite{aaronson2019gentle} to show that, the complexity of shadow tomography (assuming that the unknown state $\rho$ comes from a set $\U$) can be made $O(\sfat(\U),\log m)$.

\textbf{\emph{2}. A better bound on $\sfat(\cdot)$.} Let $\U_n$ be the class of \emph{all} $n$-qubit states.  As we mentioned earlier, Aaronson et al.~\cite{aaronson2018online} showed that $\sfat(\U_n)$ is at most $O(n)$, but clearly for a subset $\U\subseteq \U_n$ of quantum states it is possible that $\sfat(\U)\ll \sfat(\U_n)$. In this direction, using techniques from quantum random access codes (which was also used before in the works of~\cite{aaronson2007learnability,aaronson2018online,nayak,ambainis2002dense}) we first give a general upper bound on the sequential fat shattering dimension of a class of quantum states in terms of Holevo information of an ensemble. An immediate consequence of this result is a class of states for which $\sfat(\cdot)$ is much smaller than $n$. Consider the set $\U$ of ``$k$-juntas",\footnote{$k$-juntas are well-studied in computational learning theory, wherein a Boolean function on $n$ bits is a $k$-junta if it depends on an \emph{unknown} subset of $k$ input bits.} i.e.,      each $n$-qubit state lives in the same \emph{unknown} $k$-dimensional subspace. In this case it is not hard to see that Holevo information of this ensemble is at most $\log k$, which improves upon the trivial upper bound of $n$ on $\sfat(\U)$. We discuss more such classes of states~below.

\textbf{\emph{3}. Relations to Shannon theory.}  Another intriguing connection we develop in this work is between quantum learning theory and Shannon theory. This connection is already implicit from the previous point since it relates the sequential fat shattering dimension (a well-studied notion in learning theory) with Holevo information of an ensemble (which is well-studied in quantum information theory). We now establish the following: let $\U$ again be a class of quantum states and let $\mathcal{N}$ be a quantum channel. Let $\U'=\{\mathcal{N}(U):U\in \U\}$ 
be the set of states obtained after passing through quantum channel $\mathcal{N}$. Suppose the goal is to learn $\mathcal{U}'$ (i.e.,      to learn a class of states that have passed through a noisy channel $\mathcal{N}$, for example if the state preparation channel is noisy). This connects to the question of experimental learning of quantum states, i.e., can we learn states prepared using a noisy quantum device and as a by product learn the unknown noise in the quantum device.

In this case we show that $\sfat(\U') \leq C(\mathcal{N})$, i.e., the sequential fat shattering dimension is upper bounded by the classical capacity of $\mathcal{N}$. Since we have shown that $\sfat(\cdot)$ is an important parameter in many learning models, this immediately implies that the complexity of learning the class of states $\U'$ is at most the channel capacity of $\mathcal{N}$. We now give a few consequences of this result. Consider states subject to depolarizing and Pauli noise, two commonly-used noise models. For these $n$-qubit channels, channel capacity is $n-\Delta$ where $\Delta$ is an error term depending on the channel parameters~\cite{king,siudzinska2019regularized,siudzinska2020classical}. 
Hence, for extremely noisy channels, for example when $\Delta=n-o(n)$, our new bound on $\sfat(\cdot)$ is $o(n)$ which is a significant improvement over $n$. We also consider quantum learning of Gaussian states. Let $S$ be the set of Gaussian states with finite average energy $E$. Considering the pure-loss Bosonic channel with transmissivity $1$, this enables us to bound $\sfat(S) \leq O(\log E)$~\cite{giovanetti2004,giovannetti_2014}. 
Observe that previous bounds would yield $\sfat(S)<\infty$, since these states are infinite-dimensional. As far as we are aware, this is the first work to consider learnability of continuous-variable states. These connections give the Shannon-theoretic notions of channel capacity and Holevo information, a learning-theoretic interpretation, and our results can be seen as an important interdisciplinary bridge between these~fields. 

\textbf{\emph{4}. Classical contribution.} Although our main results above have been stated in terms of learning quantum states, our explicit theorem statements below are in terms of learning real-valued functions with imprecise feedback on the examples. As far as we are aware, even \emph{classically} establishing equivalences between online learning, stability and approximate differential privacy for \emph{real-valued} functions with {\em precise} feedback was only recently explored in the work of~\cite{jung2020equivalence} (which we compare against our work in the previous section), and in our work we look at imprecise feedback. Indeed, learning $n$-qubit quantum states over an orthogonal basis of $n$-qubit quantum measurements, $\mathcal{M}$, is equivalent to learning -- with imprecise adversarial feedback -- an arbitrary real-valued function in the class $\mathcal{D} = \{f:\mathcal{X} \rightarrow [0,1]\}$, for $\mathcal{X} = \mathcal{M}$: there is a one-to-one mapping between the set of all quantum states and real-valued functions on $\mathcal{M}$, i.e., for every $\sigma$, one can clearly associate a function $f_\sigma:\M\rightarrow [0,1]$ defined as $f_\sigma(M)=\Tr(M\sigma)$ and for the converse direction, given an arbitrary $c:\M\rightarrow [0,1]$, one can find a density matrix $\sigma$ for which $c(M)=\Tr(M\sigma)$ for all $M\in \M$ (and this uses the orthogonality of $\M$ crucially). Hence, if one can learn $\mathcal{D}$ when we fix the $\X$ to be over an arbitrary orthogonal basis of $2$-outcome measurements then one can learn the class of quantum states~$\Cc$, and the converse is also true. All our main theorems are stated for the general case of learning $\mathcal{D}$ for arbitrary~$\X$. 

\paragraph{Open questions.} 
We now conclude with a few concrete questions. 
\begin{enumerate}
    \item In this work we work in the $\PAC$ setting where there is an unknown concept from the class labelling the training set. Do all these equivalences also work in the \emph{agnostic} setting where there might \emph{not} be a true concept labelling the training data? The agnostic model of learning is a way to model noise which is relevant when experimentally  learning quantum states.
     \item  In a very recent work, Ghazi et al.~\cite{ghazi:sampleproperPAC} improved upon the result of Bun et al.~\cite{bun2020equivalence} by showing that a \emph{polynomial} blow-up in sample complexity suffices in going from online learning to differential privacy, which is exponentially better than the result of Bun et al.~\cite{bun2020equivalence}. Can we improve the complexity in this work using techniques from Ghazi et al.~\cite{ghazi:sampleproperPAC}?
    \item Bun showed~\cite{bun2020computational} that  the equivalence between private learning and online learning cannot be made computationally efficient (even with polynomial sample complexity) assuming the existence of one-way functions, do these also extend to the quantum setting? 
    \item  Our work shows interesting classes of states for which we can improve the complexity of gentle shadow tomography. Could we make an analogous statement for the recent improved shadow tomography procedure of~\cite{buadescu2020improved}?   Furthermore, could we further get rid of the $\sfat(\cdot)$ dependence in the sample complexity of shadow tomography\footnote{We remark that in our setting, if $\Cc$ is the class of all quantum states then $\sfat(\Cc)=n$, so we get the same complexity as Aaronson~\cite{aaronson:shadow,aaronson2019gentle}.}?
    Additionally, can we also improve standard tomography problem? 
    \item Our $\RSOA$ algorithm is time-inefficient since it compute $\sfat(\cdot)$ of arbitrary classes of states. Is there a time-efficient quantum online learning algorithm for an interesting  clas  of states? 
    \item Is there a quantum algorithm that improves the complexity of $\RSOA$?  For the case of Boolean functions,  Kothari~\cite{kothari:oracle} showed how to use quantum techniques to improve the classical halving algorithm (which is the precursor to the Standard Optimal Algorithm). Can a similar technique be applied to our $\RSOA$ algorithm?
    \item Classically, Kasiviswanathan~\cite{kasiviswanathan2011can} established connections between statistical query learning and local differential privacy. Do these connections extend also to the quantum regime, using the recently defined notion of quantum statistical query learning~\cite{quantumSQ}?
   
\end{enumerate}

\paragraph{Acknowledgements.} We thank Mark Bun for various clarifications and also providing us a proof of Claim~\ref{claim}. SA was partially supported by the IBM Research Frontiers Institute and acknowledges support from the
Army Research Laboratory, the Army Research Office under grant number W911NF-20-1-0014. YQ was supported by the Stanford QFARM fellowship and an NUS Overseas Graduate Scholarship. JS and SA acknowledge support from the IBM Research Frontiers~Institute. 

\section{Preliminaries}
\label{sec:prelim}

\paragraph{Notation.} Throughout this paper we will use the following notation. We let $\X$ be the input domain of real-valued functions (eventually when instantiating to quantum learning, we will let $\X$ be a set of 2-outcome measurements denoted by $\M$). We will let $\Cc$ be a concept class of real valued functions, i.e,. $\Cc\subseteq \{f:\X\rightarrow [0,1]\}$ and let $\Hi$ be a \emph{collection} of concept classes $\Cc$. For a distribution $D:\X\rightarrow [0,1]$,  two functions $h,c:\X\rightarrow [0,1]$ and a distance parameter $r\in [0,1]$, we define loss as
\begin{equation}
    \Loss_{D}(h,c,r) := \Pr_{x \sim D} \big[|h(x)-c(x)| > r\big]. 
\end{equation}

\paragraph{Quantum learning setting.} While we are interested in the quantum learning setting -- learning $n$-qubit quantum states in the class $\U$ over an orthogonal basis of $n$-qubit quantum measurements, $\mathcal{M}$ -- our results apply more generally to learning an arbitrary real-valued function class $\Cc= \{f:\mathcal{X} \rightarrow [0,1]\}$ with imprecise adversarial feedback. Therefore the learning models we introduce, and our theorems in the rest of this paper, will be for the more general real-valued~setting.

For $\mathcal{X} = \mathcal{M}$, these two problems are equivalent: there is a one-to-one mapping between the set of all quantum states and real-valued functions on $\mathcal{M}$, i.e., for every $\sigma$, one can clearly associate a function $f_\sigma:\M\rightarrow [0,1]$ defined as $f_\sigma(M)=\Tr(M\sigma)$ and for the converse direction, given an arbitrary $c:\M\rightarrow [0,1]$ which is the learner's hypothesis function, one can find a density matrix $\sigma$ for which $c(M)=\Tr(M\sigma)$ for all $M\in \M$ (and this uses the orthogonality of $\M$ crucially). Hence, if one can learn $\Cc$ for $\X = \M$ then one can learn the class of quantum states~$\U$, and the converse is also true.  When~$\U$ is a subset of the set of all $n$-qubit states, the learner we construct is an improper learner, i.e., it could output a density matrix $\sigma$ not in $\U$, which nevertheless is useful for prediction. If it is not important that the hypothesis function corresponds to an actual density matrix, it is not necessary to restrict the measurements to come from an orthogonal basis.  

\subsection{Learning models of interest}\label{subsec:learning_models}

\paragraph{$\PAC$ learning.}
We first introduce the $\PAC$ learning model for the real-valued concept classes.
\begin{definition}[$\PAC$ learning\label{def:RPAC}]
Let $\alpha,\zeta \in [0,1]$. An algorithm $\A$ $(\zeta,\alpha)$-$\PAC$ learns $\Cc$ with sample complexity $m$ if the following holds: for every $c\in \Cc$, and distribution $D:\X\rightarrow [0,1]$, 
given $m$ labelled examples
$\{(x_i,\widehat{c}(x_i))\}_{i=1}^m$ where each $x_i\sim D$ and $|c(x_i)-\widehat{c}(x_i)|\leq \zeta/5$, then with probability at least $3/4$ (over random examples and randomness of $\A$) outputs a hypothesis $h$ satisfying\footnote{An alternative definition of the $\PAC$ model of learning is the following: a learner obtains $(x_i,b)$ where $b\in \01$ satisfies $\Pr[b=1]=c(x_i)$. Both these models are equivalent up to poly-logarithmic factors.}
\begin{equation}
    \Pr_{y\sim D}\big[|c(y)-h(y)|\geq \zeta\big] \leq \alpha.
\end{equation}
\end{definition}
\noindent We remark that in the definition above, we assume the success probability of the algorithm is  $3/4$ for notational simplicity. With an overhead of $O(\log(1/\beta))$, we can boost $3/4$ to $1-\beta$ using standard techniques as mentioned in~\cite{izdebski2020improved}. 

\paragraph{Online learning}
Let us now introduce the online learning setting in the form of a game between two players: the learner and the adversary. As always, we shall be concerned with learning real-valued concept classes $\Cc := \{f:\mathcal{X} \rightarrow [0,1]\}$ and we let the target function be $c \in \Cc$. In the rest of this paper, we will use the term ``online learning" to refer to {\em improper} online learning, also known in the literature as online {\em prediction}, where the learner's objective is to make predictions for $c(x)$ given some point $x\in \mathcal{X}$, and it may do so using a hypothesis function $f(x)$ not necessarily in $\Cc$. Importantly, we also depart from the real-valued online learning literature in allowing the adversary to be imprecise; that is, for the adversary to respond to the learner with feedback that is $\eps$-away from the true value (this is made more precise below). This generalization allows for the case when the feedback is generated by a randomized algorithm with approximation guarantees, a statistical sample, or a physical measurement. 

The following setting, which we also call the {\em strong feedback} setting, was introduced by~\cite{aaronson2018online} to model online learning of quantum states. A protocol in this setting is a $T$-round procedure: at the $t$-th round,
\begin{enumerate}
    \item Adversary provides input point in the domain: $x_t \in \mathcal{X}$.
    \item Learner has a local prediction function $f_t$ which may not necessarily be in $\Cc$, and 
    predicts $\hat{y}_t = f_t(x_t) \in [0,1]$. 
    \item Adversary provides strong feedback $\widehat{c}(x_{t}) \in [0,1]$ 
    satisfying $|\widehat{c}(x_t)- c(x_t)| < \eps$.
    \item Learner suffers loss $\left|\hat{y}_t - c(x_t)\right|.$
\end{enumerate}
 
At the end of $T$ rounds, the learner has computed a function $f_{T+1}$, which functions as its prediction rule. If the learner is such that $f_{T+1}$ is not guaranteed to be in $\Cc$, we call the learner an `improper learner'. Such a learner can, however, still make predictions $f_{T+1}(x)$ on any given input $x \in \mathcal{X}$. Alternatively, we could also require that the learner be `proper', that is, it must output some $f_{T+1} \in \Cc$. Generally, the goal of the learner is either to make as few prediction mistakes as possible within $T$ rounds (where a `mistake' is defined as $|f_t(x_t)- c(x_t)| > \eps$, to be discussed more below); or to minimize {\em regret} for a given notion of loss, which is the total loss of its predictions compared to the loss of the best possible prediction function that could be found with perfect foresight. The former, `mistake-bound' setting is the one relevant to quantum states, so we focus on that for the rest of this paper. 
 
Some variants of our strong feedback setting could also be considered, and we now explain how they are related to our setting. Firstly,~\cite{rakhlin2010online} and~\cite{jung2020equivalence} consider an alternative setting for online prediction of real-valued functions that differs from ours in step (3). There, the adversary's feedback is $c(x_t)$ itself and is infinitely precise; to recover that setting from ours, we merely set $\eps=0$ in step (3). Since in our setting we allow $\eps$ arbitrary, we accommodate the possibility of a precision-limited adversary, for instance if the adversary's feedback comes from some estimation process or physical measurement. A second alternative setting is where the adversary only commits to providing {\em weak feedback}: $\widehat{c}(x_t)=0$ if $|\hat{y}_t- c(x_t)| < \varepsilon$ and  $\widehat{c}(x_t)=1$ otherwise. Additionally, if the latter is true, the adversary specifies if  $c(x_t) > \hat{y}_t + \varepsilon$, or $c(x_t) < \hat{y}_t - \varepsilon$ to the learner. We have termed this `weak feedback' because it contains only two bits of information, whereas for the strong feedback setting considered above, the feedback contains $O(\log(1/\eps))$ bits of information.\footnote{That is to say, a learner that works in the strong feedback setting can also work in the weak feedback setting, by mounting a binary search of the range $[0,1]$ to obtain for itself an $\eps$-approximation of strong feedback at every round. Conversely, a learner that works for the weak feedback setting also works in the strong feedback setting, by throwing away some information in the strong feedback.}

\paragraph{Mistake bound for online learning.} 
We now introduce the notion of `mistake bound' of an online learner. Before defining the model, we first define an \emph{$\varepsilon$-mistake} at step $(3)$ of the the $T$-step procedure we mentioned above.  
 \begin{definition}[$\eps$-mistake\label{def:mistake}]
Let the target concept be $c$. At a given round, let the input point be $x_t$ and the learner's guess be $\hat{y}_t$. The learner has made a mistake if 
 $     |\hat{y}_t- c(x_t)|\geq \varepsilon$.
 \end{definition}
We now define the mistake-bound model of online learning. 

\begin{definition}[Mistake bound]\label{def:mistake_bound} Let $\A$ be an online learning algorithm for class $\Cc$. Given any sequence $S=\left(x_{1}, \widehat{c}\left(x_{1}\right)\right), \ldots,\left(x_{T}, \widehat{c}\left(x_{T}\right)\right),$ where $T$ is any integer, $c \in \Cc$ and $\widehat{c}$ is the feedback of the online learner on point $x_i$.  Let $M_{\A}(S)$ be the number of mistakes $A$ makes on the sequence $S$. 

We define the \emph{mistake bound of learner $\A$ (for $\Cc$)} as $\max_{S} M_\A(S)$ where $S$ is a sequence of the above form. We say that class $\Cc$ is online learnable if there exists an algorithm $A$ for which $M_{\A}(\Cc) \leq B<\infty$. We further define the \emph{mistake bound of a concept class} as $M(\Cc) := \min_{\A} M_{\A}(\Cc)$ where the minimization is over all valid online learners $A$ for $\Cc$. 
\end{definition}

The mistake bound of class $\Cc$, $M(\Cc)$ is one way to measure the online learnability of $\Cc$. For learning Boolean function classes,~\cite{littlestone} showed that this bound gives an operational interpretation to the Littlestone dimension of the function class: $\min_{\A} M_{\A}(\Cc) = \textsf{Ldim}(\Cc)$. For showing that there exists $\mathcal{A}$ such that $M_{\A}(\Cc) \leq \textsf{Ldim}(\Cc)$, Littlestone constructed a generic algorithm -- the {\em Standard Optimal Algorithm} -- to learn any class $\Cc$ that makes at most $\textsf{Ldim}(\Cc)$-many mistakes on any sequence of examples.  The mistake-bounded online learning model outlined above for quantum states recovers the `online learning of quantum states' model proposed by~\cite{aaronson2018online}. Aaronson et al.'s work~\cite{aaronson2018online} focuses on regret bounds for online learning and here we focus on online learning with bounded mistakes. While this can be viewed as a special case of bounding regret (with an indicator loss function), the mistake-bound viewpoint opens up the connection to other models of learning, as we will see in the rest of this paper.

\subsection{Other tools of interest}

\subsubsection{Differentially-private learning}
The task of designing randomized algorithms with privacy guarantees has attracted much attention classically with the motivation of preserving user privacy~\cite{dwork2014}. Below we formally introduce \emph{differential privacy}, one way of formalizing privacy. 
 Let $\A$ be a learning algorithm. Let $S$ be a sample set consisting of labelled examples $\{(x_i,\ell_i)\}_{i \in [n]}$ where $x_i \in \mathcal{X}, \ell_i \in [0,1]$, that is fed to a learning algorithm $\A$. We say two sample sets $S,S'$ are {\em neighboring} if there exists $i \in [n]$ such that $(x_i,\ell_i) \neq (x_i',\ell_i')$ and for all $j \neq i$ it holds that $(x_j,\ell_j)=(x_j',\ell_j')$. 
Additionally, we define $(\eps,\delta)$-indistinguishability of probability distributions: for $a, b, \varepsilon, \delta \in[0,1]$ let $a \approx_{\varepsilon, \delta} b$ denote the statement
$
a \leq e^{\varepsilon} b+\delta$ and $b \leq e^{\varepsilon} a+\delta$.
We say that two probability distributions $p, q$ are $(\varepsilon, \delta)$-indistinguishable if $p(E) \approx_{\varepsilon, \delta} q(E)$ for every event $E$.

\begin{definition}[Differentially-private learning]
\label{def:privatereallearner}
A randomized algorithm
\[
\A:(\mathcal{X} \times [0,1])^{n} \rightarrow [0,1]^{X}
\]
is $(\varepsilon, \delta)$-differentially-private if for every two neighboring examples $S, S^{\prime} \in (X \times [0,1])^{n}$, the output distributions $\A(S)$ and $\A\left(S^{\prime}\right)$ are $(\varepsilon, \delta)$-indistinguishable. 
\end{definition}

\begin{definition}[Differentially-private $\PAC$ learning\label{def:PRPAC}]
Let $\Cc\subseteq \{f:\X\rightarrow [0,1]\}$ be a concept class. Let $\zeta, \alpha\in [0,1]$ be accuracy parameters and $\varepsilon,\delta$ be privacy parameters. We say $\Cc$ can be learned with sample complexity $m(\zeta,\alpha,\varepsilon,\delta)$ in a private $\PAC$ manner if there exists an algorithm $\A$ that satisfies the following:
\begin{itemize}
\item \textbf{$\PAC$ learner} --- Algorithm $\A$ is a $(\zeta, \alpha)$-$\PAC$ learner for $\mathcal{C}$ with sample size $m$ (as formulated in Definition~\ref{def:RPAC}).
\item \textbf{Privacy} --- Algorithm $\A$ is $(\eps,\delta)$-differentially private (as formulated in Definition~\ref{def:privatereallearner}).
\end{itemize}
We shall say such a learner is $(\zeta,\alpha,\varepsilon,\delta)$-$\PPAC$. 
\end{definition}

\subsubsection{Communication complexity.}\label{subsubsec:comm}
In this section, we introduce one-way classical and quantum communication complexity. Different from the usual setting, here we consider communication protocols that compute real-valued and not just Boolean functions. 
In the one-way classical communication model, there are two parties Alice and Bob. Let $\Cc\subseteq \{f:\01^n\rightarrow [0,1]\}$ be a concept class. We consider the following task which we call $\textsf{Eval}_{\Cc}$: Alice receives a function $f\in \Cc$  and Bob receives an $x\in \X$. Alice and Bob share random bits and Alice is allowed to send classical bits to Bob, who needs to output a $\zeta$-approximation of~$f(x)$ with probability $1-\varepsilon$. We let $R_{\zeta,\eps}^{\rightarrow}(c,x)$ be the \emph{minimum} number of bits that Alice communicates to Bob, so that he can output a $\zeta$-approximation of  $f(x)$ with probability at least $1-\varepsilon$ (where the probability is taken over the randomness of Alice and Bob). Let $R_{\zeta,\eps}^{\rightarrow}(\Cc)=\max \{ R_{\zeta,\eps}^{\rightarrow}(c,x):c\in \Cc, x\in \X\}$.

We will also be interested in the quantum one-way communication model. The setting here is exactly the same as above, except that now Alice and Bob
can apply quantum unitaries locally and Alice is allowed to send qubits instead of classical bits to Bob. Like before, we let $Q_{\zeta,\eps}^{\rightarrow}(c,x)$ be the \emph{minimum} number of qubits that Alice communicates to Bob, so that he can output a $\zeta$-approximation of  $c(x)$ with probability at least $1-\varepsilon$ (where the probability is taken over the randomness of Alice and Bob). Let $Q_{\zeta,\eps}^{\rightarrow}(\Cc)=\max \{ Q_{\zeta,\eps}^{\rightarrow}(c,x):c\in \Cc, x\in \X\}$.

\subsubsection{Stability of algorithms}
An important conceptual contribution in this paper is the concept of \emph{stability} of algorithms. The notion of stability has been used in several previous works~\cite{dwork2014,bun2020equivalence,alon2019private,abernethy2019online,bousquet2019passing}.  In the context of real-valued functions we are not aware of such a definition. We naturally extend previous definitions of stability from Boolean-valued functions to real-valued functions as~follows.
\begin{definition}[Stability]\label{def:stability}
Let $\Cc\subseteq \{f:X\rightarrow [0,1]\}$ be a concept class and $\eta,\zeta \in [0,1]$. Let $\De:\X\rightarrow [0,1]$ be a distribution and $c\in \Cc$ be a target unknown concept. We say a learning algorithm $\A$ is $(T,\eta,\zeta)$-\emph{stable with respect to $D$} if: given $T$ many labelled examples $S=\{(x_i,c(x_i))\}$ when $x_i\sim \De$, there exists a hypothesis $f$ such that
$$
\Pr[\A(S) \in \T(\zeta, f)]\geq \eta,
$$
where the probability is taken over the randomness of the algorithm $\A$ and the examples $S$.
\end{definition}

It is worth noting that in the standard notion of global stability (for example the one used  in~\cite{bun2020equivalence}), we say an algorithm $\mathcal{A}$ is stable if a  {\em single } function is output by $\mathcal{A}$ with high probability. In the real-valued robust scenario, one cannot hope for similar guarantees because the adversary is allowed to be $\zeta$-off with his feedback at every round. In particular, the adversary's feedback could correspond to a different function from the target concept $c$. However, the intuition is that any adversarially-chosen alternative function cannot be ``too" far from $c$.

Inspired by the definition above we also define quantum stability as follows. 

\begin{definition}[Quantum Stability]\label{def:qstability}
Let $S$ be a class on $n$-qubit quantum states and $\eta,\zeta \in [0,1]$. Let $\De:\X\rightarrow [0,1]$ be a distribution over orthogonal $2$-outcome measurements and $\rho\in  S$ be an  unknown quantum state. We say a learning algorithm $\A$ is $(T,\eta,\zeta)$-\emph{stable with respect to $D$} if: given $T$ many labelled examples $Q=\{(E_i,\Tr(\rho E_i))\}$ when $E_i\sim \De$, there exists a quantum state $\sigma$  such that
\begin{equation}
\Pr[\calA(Q) \in \calB(\varepsilon, \sigma)]\geq \eta,
\end{equation}
where the probability is taken over the examples in $Q$ and $\calB(\varepsilon,\sigma)$ is the ball of states $\eps$-close to $\sigma$ with respect to $\X$, i.e.,  $\calB(\varepsilon,\sigma)=\{\sigma': |\Tr(E\sigma) - \Tr(E\sigma')|<\varepsilon \, \text{ for every } E \in \mathcal{X}\}$.
\end{definition}

\subsubsection{Combinatorial parameters.}
We define some combinatorial parameters used in $\PAC$ learning and online learning real-valued function classes $\{f:\mathcal{X} \rightarrow [0,1]\}$. These are the fat-shattering (for $\PAC$ learning) and sequential fat-shattering dimension (for online learning). They can be viewed as the real-valued analogs of the VC dimension and Littlestone dimension respectively for $\PAC$ learning and online learning Boolean function classes $\{f:\mathcal{X} \rightarrow \{0,1\}\}$. Below we define the combinatorial parameters for real-valued~functions.
\paragraph{Fat-Shattering dimension}
The set $\left\{x_{1}, \ldots, x_{k}\right\} \subseteq \mathcal{X}$ is $\gamma$-fat-shattered by concept class $\Cc$ if  there exists real numbers $\left\{\alpha_{1}, \ldots, \alpha_{k}\right\}\in [0,1]$ such that for all $k$-bit strings $y=(y_{1} \cdots y_{k})$ there exists a concept $f \in \Cc$ such that
\emph{if} $y_{i}=0$ then $f\left(x_{i}\right) \leq \alpha_{i}-\gamma$
 and \emph{if} $y_{i}=1$ then $f\left(x_{1}\right) \geq \alpha_{i}+\gamma$.
 
The fat-shattering dimension of $\Cc$, or $\fat_{\gamma}(\Cc)$ is the largest $k$ for which: there exists $\left\{x_{1}, \ldots ,x_{k}\right\} \in \mathcal{X}$ that is $\gamma$-fat-shattered by $\Cc$. We remark that if the functions in $\Cc$ have range $\01$ and $\gamma>0$, then $\fat_\gamma(\Cc)$ is just the standard $\textsf{VC}$ dimension.

\paragraph{Sequential Fat-Shattering dimension}
We also define an analog of the fat-shattering dimension for online learning. The presentation of this dimension closely follows~\cite{aaronson2018online}.
We say a depth-$k$ tree $T$ is an \emph{$\eps$-sequential fat-shattering tree} for $\Cc$ if  it satisfies the following:
\begin{enumerate}
    \item For every internal vertex $w\in T$, there is some domain point $x_w\in U$ and threshold $a_w \in [0,1]$ associated with $w$, and
    \item For each leaf vertex $v \in T$, there exists $ f \in \Cc$ that causes us to reach $v$ if we traverse $T$ from the root such that at any internal node $w$ we traverse the left subtree if $f\left(x_{w}\right) \leq a_{w}-\varepsilon$ and the right subtree if $f\left(x_{w}\right) \geq a_{w}+\varepsilon .$ If we view the leaf $v$ as a $k$ -bit string, the function $f$ is such that for all ancestors $u$ of $v,$ we have $f\left(x_{u}\right) \leq a_{u}-\varepsilon$ if $v_{i}=0,$ and $f\left(x_{u}\right) \geq a_{u}+\varepsilon$ if $v_{i}=1,$ when $u$ is at depth $i-1$ from the root.
   \end{enumerate}
The \emph{$\varepsilon$-sequential fat-shattering dimension} of $\Cc$, denoted $\sfat_{\eps}(\Cc)$, is the largest $k$ such that we can construct a complete depth-$k$ binary tree $T$ that is an $\eps$-sequential fat-shattering tree for $\Cc$. 
 Again, we remark that if the functions in $\Cc$ have range $\01$ and $\gamma>0$, then $\sfat_\gamma(\Cc)$ is just the standard Littlestone dimension~\cite{littlestone}.

\paragraph{Representation dimension.}
The representation dimension of concept class $\Cc$ 
roughly considers the collection of all distributions  over sets of hypothesis functions (not necessarily from the class~$\Cc$) that “cover” $\Cc$. We make this precise below. 
This dimension is known to capture the sample complexity of various models of differential private learning Boolean functions~\cite{kasiviswanathan2011can,beimel2010}.
 Because we shall be concerned with learning real-valued concept classes, we define these notions below with an additional `tolerance' parameter~$\zeta$.
 
\begin{definition}[Deterministic representation dimension $\DRD$, real-valued analog of~\cite{beimel2010}] \label{def:DRD}
Let $\Cc\subseteq \{f:\X\rightarrow [0,1]\}$ be a concept class. A class of functions $\Hi$ deterministically $(\zeta,\eps)$-represents~$\Cc$ if for every $f \in \Cc$ and every distribution $\mathcal{D}:\X\rightarrow [0,1]$, there exists $h\in \Hi$ such~that
\begin{equation}
 \Pr_{x\sim D}\big[|h(x)-f(x)| >\zeta\big] \leq \eps. \end{equation}
 The deterministic representation dimension  of $\Cc$ (abbreviated $\DRD(\Cc)$) is
 \begin{equation}
 \DRD_{\zeta,\eps}(\Cc) = \min_{\Hi} \log |\Hi|
 \end{equation}
where the minimization is over $\Hi$ that deterministically $(\zeta,\eps)$-represent $\Cc$. 
\end{definition}

\begin{definition}[Probabilistic representation dimension $\PRD$, real-valued analog of~\cite{beimel2013}] \label{def:PRD}
Let $\Cc\subseteq \{f:\X\rightarrow [0,1]\}$ be a concept class. Let $\mathscr{H}$ be a collection of concept classes of real-valued functions, and $\mathcal{P}:\mathscr{H}\rightarrow [0,1]$. We say $(\mathscr{H},\mathcal{P})$ is $(\zeta,\eps,\delta)$-representation of $\mathcal{C}$ if for every $f \in \Cc$ and distribution $D:\mathcal{X}\rightarrow [0,1]$, with probability at least $1-\delta$ (over the choice of $\Hi\sim \mathcal{P}$), there exists $h\in \mathcal{H} $ such that
\begin{equation}
    \Pr_{x\sim D}\big[|h(x)-f(x)| >\zeta\big] \leq \eps. 
\end{equation}
The probabilistic representation dimension  of $\Cc$ (abbreviated $\PRD(\Cc)$) is
\begin{equation}
\PRD_{\zeta,\eps,\delta}(\Cc) = \min_{(\mathscr{H},\mathcal{P})} \max_{\mathcal{H} \in \supp(\mathscr{H})} \log |\mathcal{H}|,
\end{equation}
where the outer minimization is over all sets $(\mathscr{H},\mathcal{P})$ of valid $ (\zeta,\eps,\delta)$-representations. 
\end{definition}

\section{Robust standard optimal algorithm and mistake bounds}
\label{sec:RSOA}
In this section, we present an algorithm that improperly online-learns a real-valued function class~$\Cc$, making at most $\sfat(\Cc)$ many mistakes (see Definition~\ref{def:mistake_bound}). This algorithm is an important tool for results in the rest of the paper. All results in this section are presented for the general case of online-learning arbitrary real-valued function classes, with imprecise adversarial feedback. Ultimately, we will use this algorithm as a subroutine for the specific setting of quantum learning. 

Our algorithm's learning setting generalizes that of \cite{rakhlin2010online} and \cite{jung2020equivalence}, who also studied online learning of real-valued and multi-class functions (i.e., functions mapping to a finite set), albeit, the former in the case of precise adversarial feedback ($\eps= 0$). \cite{jung2020equivalence} defined several extensions of the Littlestone dimension $\Ldim_{\tau}$ for $\tau \in \mathbb{Z}_+$ and showed that for learning a multi-class function class $\Cc$, $\Ldim_{\tau} < M(\Cc) < \Ldim_{2\tau}$. They also showed that for a real-valued function class $\Cc$, $\sfat(\Cc)$ is linked to the $\Ldim_{\tau}$ of a discretization of the function class, thus effectively transforming any real-valued learning problem into a multi-class learning problem. However, their approach does not work for our setting, for the following reason: if $c$ is the target real-valued function, and the true value of $c(x)$ is $\eps$-close to a boundary of some class within the discretized range, our $\eps$-imprecise adversary could choose a value of the feedback $\widehat{c}(x)$ that falls in the neighboring class. Hence the resulting multi-class learner has to deal with the adversary reporting the wrong class, which is beyond the scope of what they considered.

In Section~\ref{subsec:RSOA}, we first construct an algorithm Robust Standard Optimal Algorithm ($\RSOA$) whose mistake bound satisfies $M_{\RSOA}(\Cc) \leq \sfat(\Cc)$ for online-learning with strong feedback. In Section~\ref{subsec:RSOA_properties}, we prove some of the properties of this algorithm, which are essential for proving later results in this paper. Moreover, for online learning with weak feedback, we show that any adversary can force at least $\sfat(\Cc)$ mistakes. We cannot however make the same statement for online learning with strong feedback (this would be the real-valued analog of the relation $M(\Cc) = \Ldim(\Cc)$ proved by Littlestone for Boolean function classes). It is an open question whether we can close this gap, but for the rest of this paper, we are concerned solely with online learning with strong feedback and hence the implication $M_{\RSOA}(\Cc) \leq \sfat(\Cc)$ is~sufficient.

\subsection{Robust Standard Optimal Algorithm}\label{subsec:RSOA}
In this section, we give an algorithm to to online-learn real-valued functions with strong feedback. In order to handle subtleties caused by learning functions with output in $[0,1]$ instead of $\01$,
 we define the notion of an $\zeta$-cover. This was introduced by~\cite{rakhlin2010online} and in order to handle inaccuracies in the output of an adversary, we extend their notion to define an \emph{interleaved $\zeta$-cover}.

\begin{definition}[$\zeta$-cover and interleaved $\zeta$-cover]
 \label{def:cover}
Let $0<\zeta< 1$ be such that $1/\zeta$ is an integer. A \emph{$\zeta$-cover} of the $[0,1]$ interval is a set of non-overlapping half-open intervals (`bins') of width $\zeta$ given by $\big\{[0,\zeta),[\zeta,2\zeta),\ldots,[1-\zeta,1]\big\}$ with the midpoints
\[
\In_{\zeta}=\big\{\zeta / 2,3 \zeta / 2, \ldots, 1 -\zeta/2 \big\}
\] 
where 
 $|\In_{\zeta}| = 1/\zeta$.  Given a $\zeta$-cover $\In_{\zeta}$, the corresponding \emph{interleaved $\zeta$-cover} $\tilde{\In}_{\zeta}$ is the set of overlapping half-open intervals (`super-bins') of width $2\zeta$ (each consisting of two adjacent bins in $\In_{\zeta}$) given by $\big\{[0,2\zeta),[\zeta,3\zeta),\ldots,[1-2\zeta,1]\big\}$ with the midpoints 
\[
\tilde{\In}_{\zeta}=\big\{\zeta, 2 \zeta, \ldots, 1 -\zeta \big\}
\]
where 
$|\tilde{\In}_{\zeta}| = |\In_{\zeta}|-1$.  We denote a super-bin with midpoint $r$ as $\SB(r)$.
\end{definition}

We will also need the definition of a $\zeta$-ball.
\begin{definition}[$\zeta$-ball]
An $\zeta$-ball around an arbitrary point $x\in[0,1]$ (denoted $B(\zeta,x)$) is the open interval of radius $\zeta$ around $x$, i.e.,      $   B(\zeta, x):= (x-\zeta,x+\zeta)$
\end{definition}
As we mentioned earlier, the $\textsf{FAT}$-$\textsf{SOA}$ algorithm of~\cite{rakhlin2010online} used $\alpha$-covers to understand real-valued online learning, however, it does not suffice in the setting of quantum learning since the output of the adversary could be imprecise. To account for this, we use interleaved $\alpha$-covers defined above. 
Our learning algorithm will take advantage of the following property enjoyed by the interleaved $\alpha$-cover: the $\zeta$-ball of any point is guaranteed to be {\em entirely} contained inside some super-bin, i.e.,      for every $x\in (\zeta,1-\zeta)$,  $\alpha>2\zeta$ and $r =\argmin_{r\in \tilde{\In}_{2\zeta}} \{|x-r|\}$, we have $B(\zeta, x) \subset \SB(r)$.
Finally, we need one more notation: given a set of functions $V\subseteq \{f:\X\rightarrow [0,1]\}$, $r \in \tilde{\In}_{2\zeta}$ and $x\in \X$, define a (possibly empty) subset $V(r, x) \subseteq V$ as
\[
 V(r, x)=\big\{f \in V : f(x) \in B(2\zeta, r)\big\},
\]
i.e.,      $V(r,x)$ are the set of functions $f\in V$ for which $f(x)$ is within a $2\zeta$-ball around $r$ or $f(x)\in [r-2\zeta,r+2\zeta]$. We are now ready to present our mistake-bounded online learning algorithm for learning real-valued functions.
Our algorithm is Algorithm~\ref{algo:RSOA}.
  \begin{algorithm}[H]
		\textbf{Input:} Concept class $\Cc\subseteq \{f:\X\rightarrow [0,1]\}$, target (unknown) concept $c\in \Cc$, and $\zeta\in [0,1]$.
		\vspace{5pt}
		\textbf{Initialize}: $V_1 \gets \Cc$
		\vspace{5pt}
		\begin{algorithmic}[1]
\For{$t = 1, \ldots, T$}
\vspace{1mm}
    \State A learner receives $x_t$ and maintains set $V_t$, a set of ``surviving functions". \;
    \State For every super-bin midpoint $r\in \tilde{\In}_{2\zeta}$
    the learner computes the set of functions $V_t(r,x_t)$.
    \State A learner finds the super-bin which achieves the maximum $\sfat(\cdot)$ dimension
      $$
     R_{t}(x_t):=\left\{
    \argmax _{r \in \tilde{\In}_{2\zeta}} \sfat_{2\zeta}\left(V_{t}(r, x_t)\right)\in \tilde{\In}_{2\zeta}\right\}
    $$
    \State The learner computes the mean of the set $R_t(x_t)$, i.e.,      let 
    $$
    \hat{y}_t:=\frac{1}{\left|R_{t}(x_t)\right|} \sum_{r \in R_{t}(x_t)} r.
    $$ 
    \State The learner outputs $\hat{y}_t$ and receives feedback $\widehat{c}(x_{t})$.     
    \State Learner makes the update $V_{t+1} \leftarrow \{g \in V_t \mid g(x_t) \in B(\zeta,\widehat{c}(x_{t}))\}$ 
\EndFor
\end{algorithmic}

\textbf{Outputs:} The intermediate predictions $\hat{y}_t$ for $t\in[T]$, and a final prediction function/hypothesis which is given by $f(x):= R_{T+1}(x)$.
\caption{\textsf{Robust Standard Optimal Algorithm, $\RSOA_{\zeta}$}
}
\label{algo:RSOA}
\end{algorithm}
We first provide some intuition about this algorithm. At round $t$, the set of functions that has `survived' all previous rounds is $V_t$: in particular, $V_t$ consists of functions which are consistent with the feedback received in the previous $t-1$ iterations. Here, `consistent' means that suppose $x_1,\ldots ,x_{t-1}$ were presented to a learner previously, then, for every $g\in V_t$, $g(x_i)\in B(\zeta,\widehat{c}(x_i))$ for $i\in [t-1]$. This is clear from Line 7 of the algorithm; indeed, notice that $V_t$ either stays the same as $V_{t-1}$ or shrinks at every round. At round $t$,
 once a learner receives $x_t$, it always replies with $\hat{y}_t$ that is either $\zeta$-close to the true $c(x_t)$ else, aims to reduce $V_{t-1}$ as much as possible. In particular, for every super-bin $r\in \tilde{\In}_{2\zeta}$, the learner identifies the subset of surviving functions that map to that super-bin at $x_t$, i.e.,     $f\in V_t$ that satisfy $f(x_t)\in B(2\zeta, r)$. This forms the set $V_t(r, x_t)$. The learner then computes $\sfat_{2\zeta}$ of the set of functions $V_t(r, x_t)$ and picks out the super-bins $r\in \tilde{\In}_{2\zeta}$ that maximize this combinatorial quantity, and output the mean of their midpoints as the prediction $\hat{y}_t$. Intuitively, the parameter $\sfat(\cdot)$ serves as a surrogate metric for the number of functions mapping to a certain interval. Using $\sfat(\cdot)$ to define this prediction rule thus maximizes the number of eliminated functions for every mistake of the learner. Once it receives the feedback $\widehat{c}(x_t)$, the learner updates $V_t$ to $V_{t+1}$ and this process repeats for $T$ steps. We now list a few properties of this algorithm.

\subsection{Properties and guarantees of RSOA} \label{subsec:RSOA_properties}
\begin{lemma} 
\label{lem:props_RSOA}
 $\RSOA_{\zeta}$(denoted $\RSOA$) has the following properties: 
\begin{enumerate}
    \item $\zeta$-consistency: at the $t$-th iteration every $f\in V_t$ satisfies $|f(x_i)-\widehat{c}(x_{i})| \leq \zeta$ for  $i\in [t-1]$.
    \item Correctness: the target function $c$ is never eliminated, i.e.,       $c \in V_t$ for every $t \in [T]$.
    \item For every $t\in [T],  x \in \mathcal{X}$, any pair of points $r, r^{\prime} \in \tilde{\In}_{2\zeta}$ for which 
    \begin{equation}\label{eq:maximal_sfat}
    \sfat_{2\zeta}\left(V_t(r, x)\right) = \sfat_{2\zeta}\left(V_t(r', x)\right) = \sfat_{2\zeta}\left(V_t\right)
    \end{equation}
    also satisfies $|r-r^{\prime}|< 4\zeta$. 	Additionally for all $r\in \tilde{\In}_{2\zeta}$, $\sfat_{2\zeta }\left(V_{t}(r, x)\right) \leq \sfat_{2\zeta}\left(V_{t}\right)$.
    \item  $\RSOA$ is deterministic, 
    i.e.,      for the same sequence of inputs $(x_1,\widehat{c}(x_1)),\ldots,(x_T,\widehat{c}(x_T))$ provided by the adversary to the learner (each of which is followed by a response $\widehat{y}_1,\ldots,\widehat{y}_T$ of the learner), the $\RSOA$ algorithm produces the same function $f$.
    \end{enumerate}
\end{lemma}

\begin{proof}
The first item follows by construction.   
At the end of $i$th round, the following update is performed: $V_{i+1} \leftarrow \{g \in V_i \mid g(x) \in B(\zeta,\widehat{c}(x_{i})))\} \subseteq V_{i}$. This eliminates all functions $g$ for which $g(x_i) \notin B(\zeta, \widehat{c}(x_{i}))$ from the set $V_{i+1}$, hence all functions for which $|f(x_i)-\widehat{c}(x_{i})| > \zeta$ are eliminated. 

The second item follows trivially: by assumption $y_{t} = c(x_t)$ is in the $\zeta$-ball of $\widehat{c}(x_{t})$. Thus the target concept $c$ is never eliminated in the update $V_{t+1} \leftarrow \{g \in V_t \mid g(x) \in B(\zeta,\widehat{c}(x_{t}))\}$.
    
We now show the third item. 
 Suppose by contradiction, there is a pair $r, r^{\prime} \in \tilde{\In}_{2\zeta}$ such~that 
\[
\sfat_{2\zeta}\left(V_t(r, x)\right) = \sfat_{2\zeta}\left(V_t(r', x)\right) = \sfat_{2\zeta}\left(V_t\right)
\] 
and $|r-r^{\prime}|> 4\zeta$. Let  $\sfat_{2\zeta}\left(V_t\right)=d$.  Without loss of generality, we assume $r>r'$. Then let $s = (r+r')/2$. Clearly, for every $f\in V_t(r,x)$ we have $f(x) \geq s+ \zeta$ and  $g\in V_t(r',x)$  we have $g(x) \leq s- \zeta$. This means that, given a sequential fat-shattering tree of depth $d$ for $V_t(r,x)$, and the tree also of depth $d$ for $V_t(r',x)$, we may join them together by adding a root node with the label $x$ and the threshold $s$, and this new tree of depth $d+1$ is sequentially fat-shattered by $V_t(r,x) \cup V_t(r',x)$ and hence by $V_t$ (which is a superset). This contradicts the assumption that  $\sfat_{2\zeta}(V_t)=d$, because by definition of $\sfat(\cdot)$ dimension, $d$ is the depth of the \emph{deepest} tree for the functions in $V_t$. The ``additionally" part follows immediately because $V_t(r,x) \subseteq V_t$.

The final item of the lemma is clear because steps $3$ to $7$ in the $\RSOA$ algorithm are deterministic and involve no randomness from a learner.
\end{proof}

Having established these properties, are now ready to prove our main theorem bounding the maximum number of prediction mistakes that $\RSOA$ makes.

\begin{theorem}[$\RSOA$ mistake bound]
\label{thm:RSOA}
Let $\Cc\subseteq \{f:\X\rightarrow [0,1]\}$ be a concept class and $\zeta>0$. Given the setting of online learning with strong feedback, i.e.,      at every round $t \in [T]$, the feedback $\widehat{c}(x_{t})$ is $\zeta$-close to the true value $|c(x_t)-\widehat{c}(x_{t})| \leq \zeta$, $\RSOA_{\zeta}$ (described in Algorithm~\ref{algo:RSOA})  is such that, for every $T$, the algorithm makes a predictions $\hat{y}_t$ satisfying
\[
\sum_{t=1}^{T} \id\big[\left|\hat{y}_t - c(x_t)\right|> 5\zeta \big] \leq \sfat_{2\zeta}{(\Cc)}
\]
\end{theorem}

\begin{proof}
The intuition is that whenever the learner makes a mistake, functions are eliminated from the `surviving set', such that $\sfat(\cdot)$ of the remaining functions decreases by $1$. Since the true function $c$ is never eliminated from $V_t$, and the $\sfat(\cdot)$ dimension of a set consisting of a single function is $0$, no more than $\sfat(\cdot)$ mistakes can be made. 

First observe that, whenever the algorithm makes a mistake, i.e.,      $|\hat{y}_t - c(x_t)| > 5\zeta$, it also follows that $|\hat{y}_t - \widehat{c}(x_{t})| > 4\zeta$ because $\widehat{c}(x_{t})$ is an $\zeta$-approximation of $c(x_t)$. Below we show that on every round where $|\hat{y}_t - \widehat{c}(x_{t})| > 4\zeta$, $\sfat(V_{t+1})\leq \sfat(V_t)-1$. Together with property~2 of Lemma~\ref{lem:props_RSOA} and the fact that $V_1=\Cc$ this already implies that no more than $\sfat(\Cc)$ mistakes are made by $\RSOA$.  

Suppose $|\hat{y}_t - \widehat{c}(x_{t})| > 4\zeta$. Fix $t$ and $x_t$. Observe that by property 3  Eq.~\eqref{eq:maximal_sfat} (in Lemma~\ref{lem:props_RSOA}) there are at most three super-bins whose midpoints~$r$ satisfy $\sfat_{2\zeta}\left(V_t(r, x)\right) = \sfat_{2\zeta}\left(V_t\right)$, i.e.,      between $0$ and $3$ super-bins achieve the upper-bound on $\sfat(\cdot)$ at each round, which we now call $\UB_t:=\sfat_{2\zeta}(V_t)$.  We now analyze each of four cases for the number of upper-bound-achieving super-bins. 

\textbf{Case 1}: $\sfat_{2\zeta}(V_t(r,x_t)) < \UB_t$ for every $r\in \tilde{\In}_{2\zeta}$, i.e.,      no super-bins achieve $\UB_t$. Every update of $V_t$ updates it to the functions within some $\zeta$-ball, $\bigcirc := B(\zeta,\widehat{c}(x_{t}))$. Observe that  $\bigcirc$ is entirely contained within some super-bin, call it $\SB$ (note that even if $\widehat{c}_t$ is at the boundary of two super-bins, it would still be inside the super-bin that is in-between the two, by definition of the interleaved $\zeta$-cover). Hence, $\sfat(\bigcirc) \leq \sfat(\SB) <  \UB_t$ where the second inequality is by the assumption of the case.

\textbf{Case 2}: There exists exactly one $r\in \tilde{\In}_{2\zeta}$ such that 
$$
\sfat_{2\zeta}(V_t(r,x_t)) =  \UB_t,
$$
i.e.,      exactly one super-bin (centered at $r=2k\zeta$ for some $k\in \Z_+$) achieves $\UB_t$, let's call this $\SB^{\ast}=[2(k-1)\zeta, 2(k+1)\zeta)$. Since the super-bin's midpoint is at some bin boundary, the prediction is $\hat{y}_t = 2k\zeta$. Similar to the previous case, the update step retains only the functions in some $\bigcirc := B(\zeta,\widehat{c}(x_{t}))$. However, since $|\hat{y}_t - \widehat{c}(x_{t})| > 4\zeta$, we either have $\widehat{c}(x_{t})<2(k-2)\zeta$ or $\widehat{c}(x_{t})>2(k+2)\zeta$. $\bigcirc$, therefore, is entirely contained within some super-bin $\SB \neq \SB^{\ast}$. Since there is only one maximizing super-bin $\SB^{\ast}$, we have $\sfat(\bigcirc) \leq \sfat(\SB) < \sfat(\SB^{\ast}) =  \UB_t$. 

\textbf{Case 3}: There exists $r_1, r_2\in \tilde{\In}_{2\zeta}$ such that  
$$
\sfat_{2\zeta}(V_t(r_1,x_t)) = \sfat_{2\zeta}(V_t(r_2,x_t)) = \UB_t,
$$
i.e.,      two super-bins (centered at $r_1$, $r_2$ respectively) achieve $\UB_t$, call them $\SB_1^{\ast}, \SB_2^{\ast}$. Using Property 3 of Lemma~\ref{lem:props_RSOA}, these two super-bins must either be touching at a boundary (hence $\hat{y}_t = 2k\zeta$ where $\SB_1^{\ast}= [2k\zeta, 2(k+2)\zeta)$, $\SB_2^{\ast}= [2(k-2)\zeta, 2k\zeta)$) or intersecting at one bin (hence $\hat{y}_t = (2k+1)\zeta$ where $\SB_1^{\ast}= [2k\zeta, 2(k+2)\zeta)$, $\SB_2^{\ast}= [2(k-1)\zeta, 2(k+1)\zeta)$). In the former case, $\widehat{c}(x_{t})<2(k-2)\zeta$ or $\widehat{c}(x_{t})>2(k+2)\zeta$ and thus neither $\SB_1^{\ast}$ nor $\SB_2^{\ast}$ entirely contains $\bigcirc$, though there is some super-bin that does. In the latter case, $\widehat{c}(x_{t})<(2k-3)\zeta$ or $\widehat{c}(x_{t})> (2k+5)\zeta$ and thus neither $\SB_1^{\ast}$ nor $\SB_2^{\ast}$ entirely contains $\bigcirc$, though there is some super-bin that does. Identical reasoning to the previous two cases shows that the update thus decreases $\sfat(\cdot)$ on the remaining functions. 

\textbf{Case 4}: There exists $ r_1, r_2, r_3\in \tilde{\In}_{2\zeta}$ such that  $$
\sfat_{2\zeta}(V_t(r_1,x_t)) = \sfat_{2\zeta}(V_t(r_2,x_t)) = \sfat_{2\zeta}(V_t(r_3,x_t)) = \UB_t,
$$
i.e.,      three super-bins (centered at $r_1$, $r_2$, $r_3$ respectively) achieve $\UB_t$. Call them $\SB_1^{\ast}, \SB_2^{\ast}, \SB_3^{\ast}$. By Property 3 of Lemma~\ref{lem:props_RSOA}, there is only one configuration these three super-bins could be in, namely two super-bins have to be touching at a boundary, with the last super-bin straddling them: $\SB_1^{\ast}= [2k\zeta, 2(k+2)\zeta)$, $\SB_2^{\ast}= [2(k-1)\zeta, 2(k+1)\zeta)$, $\SB_3^{\ast}= [2(k-2)\zeta, 2k\zeta)$. Then $\hat{y}_t = 2k\zeta$ and $\widehat{c}(x_{t})<2(k-2)\zeta$ or $a-t > 2(k+2)\zeta$. None of $\SB_1^{\ast}, \SB_2^{\ast}, \SB_3^{\ast}$ entirely contains $\bigcirc$, though there is some super-bin that does, and identical reasoning to the previous three cases shows that the update thus decreases $\sfat(\cdot)$ on the remaining functions. 
\end{proof}

Theorem~\ref{thm:RSOA} says that the $\RSOA$ algorithm for a concept class $\Cc$ in the strong feedback model, makes at most $\sfat(\Cc)$ mistakes.  This is also the setting in the rest of the paper as well as most of the real-valued online learning literature. A natural question is, can we make fewer mistakes than the $\RSOA$ algorithm? Below we consider the \emph{weak} feedback model of online learning and show
no learner can do better than making $\sfat(\cdot)$ mistakes. An interesting open question is, can we even improve the lower bound in the theorem below for the strong feedback model setting? 

 \begin{theorem}\label{thm:qonline=sfat}
Let $\zeta\in [0,1]$ and $\Cc\subseteq \{f:\X\rightarrow [0,1]\}$. Every online learner $\mathcal{A}$ (in the weak feedback setting) for the class $\Cc$, satisfies $M_{\mathcal{A}}(\Cc) \geq~\sfat_{\zeta}(\Cc)$.
\end{theorem}

\begin{proof}
We construct an adversary that can always force at least $\sfat(\Cc)$ mistakes in the weak model of learning (where the adversary only gives two bits of feedback to the learner). To do so, the adversary traverses the $\zeta$-fat-shattered tree starting at the root node, at every round interacting with the learner based on the information at the current node, {\em always} claiming the learner made a mistake, and then moving to one of the two daughter nodes. In particular, the interaction at node $v$ of the tree, which is associated with $(x_v,a_v)$, is as follows: The adversary gives the learner the point $x_{v}$. If the learner predicts $\hat{y}_t< a_v$, claim the learner is wrong and go to the right daughter node, thus committing the adversary to the subset of functions $f\in \Cc$ such that $f(x_v) \geq a_v +\zeta$. Go to the opposite node if the learner predicts $\hat{y}_t \geq a_v$. After $\sfat_{\zeta}(\Cc)$ rounds, the adversary will have reached a leaf node. At this point, by the definition of the $\sfat(\cdot)$ tree, there is at least one function consistent with all previous commitments of the adversary. This becomes the target function, which the adversary then commits to in the first place.
Since the depth of the tree is by definition $\sfat_{\zeta}(\Cc)$, the learner will have made $\sfat_{\zeta}(\Cc)$ mistakes by the time the adversary reaches a leaf and has to commit to a function.
\end{proof}

\section{Online learning implies stability}
\label{sec:onlineimpliesDP}

In this section we show that online learnability of a real-valued function class implies that there exists a real-valued $\DP$ $\PAC$ learner for the same class. More precisely, we will assume that the $\sfat(\cdot)$ dimension of the function class is bounded (which implies its online learnability, as discussed in Section~\ref{sec:RSOA}); then we will explicitly describe an algorithm that uses this learner to learn in a globally-stable manner. 

This, however, is only half of the implication shown in \cite{bun2020equivalence}. There, they go one step further and turn their stable learner into an approximately $\DP$ $\PAC$ learner, concluding overall that online learning implies approximate $\DP$ $\PAC$ learning. Supposing we could prove the same for our learning model, then combining this with the implication shown in Section \ref{sec:pureDPimpliesonline} (that pure $\DP$ $\PAC$ learning implies online learning) would make for almost a complete chain of implications starting at pure $\DP$ $\PAC$ learning, implying online learning, and finally implying approximate $\DP$ $\PAC$ learning. However, in the second half of this section, we use an argument from fingerprinting codes to show that the transformation in \cite{bun2020equivalence} from a stable learner to a $\DP$ $\PAC$ learner does not work with the stability guarantees we obtain for our real-valued learning setting.

We will use the following notation throughout this section. Let $\Cc\subseteq \{f:\X\rightarrow [0,1]\}$ be a concept class and $c\in \Cc$ be a target concept. Let $D:\X\rightarrow [0,1]$ be a distribution. In a slight abuse of notation, we use the notation $(x,\widehat{c}(x)) \sim D$ to mean that $x$ is drawn from the distribution $D$ and $\widehat{c}(x)$ satisfies 
$|\widehat{c}(x)- c(x)| < \zeta$. Also, we say $B \sim D^m$ to mean that a learner receives $m$ such examples $\{(x_i,\widehat{c}(x_i))\}_{i=1}^m$. We say that the learner has made a \emph{mistake} on input $x$ if he has made a $5\zeta$-mistake (refer to Definition~\ref{def:mistake}). Finally, because we are concerned with {\em real-valued} learning, functions in the vicinity of the target function are considered ``close enough" as hypotheses, and so we will make use of the following notion of {\em function ball}:

\begin{definition}[Function ball of radius $r$ around $c$\label{def:function_ball}]
Given a set of functions $\Hi \subseteq \{f:\X\rightarrow [0,1]\}$, a function ball of radius $r$ around  $c\in \Hi$ is the set of all functions $f\in \Hi$ such that
\begin{equation}
     |f(x)-c(x)|<r \quad \text{ for every } x\in \X,
\end{equation}
and we denote such a function ball by $\T(r,c)$.\footnote{The symbol $\T$ stands for `tube' since for a member of the function ball, closeness to $c$ must be satisfied at not just a single point but all points in the domain. We omit mentioning the function class $\Cc$, which is usually taken to be $\mathcal{R}$, the set of all functions output by $\RSOA$. Because $\RSOA$ is an improper learner, $\mathcal{R}$ is not the same as~$\Cc$.}   
Moreover, for a set of functions $\E=\{f_1,\ldots,f_k\}$, we let $\T(r,\E)=\cup_{i=1}^k \T(r,f_i)$. 
\end{definition}

In Section~\ref{sec:onlineimpliesstab}, we prove that given a mistake-bounded online learner, there exists a stable learner. In Section~\ref{sec:stabdoesntimplyapproxDP}, we prove that stability does not, in turn, imply approximate $\DP$ learning using the transformation of \cite{bun2020equivalence}, without a domain size dependence in the sample complexity. In Section~\ref{sec:quantumimplications}, we turn our attention to how our results apply to learning quantum~states.

\subsection{Online learning implies stability}
\label{sec:onlineimpliesstab}

In this section we prove the following theorem:
\begin{theorem}
\label{thm:onlineimpliesstab}
Let $\alpha,\zeta \in [0,1]$. Let $\Cc\subseteq \{f:\X \rightarrow [0,1]\}$ be a concept class with $\sfat_{2\zeta}(\Cc)=d$. Let $D:\mathcal{X}\rightarrow [0,1]$ be a  distribution and let $S=\{(x_i, \widehat{c}(x_i))\}$ be a set of 
$$
T=O \left(\zeta^{-d}\cdot \frac{d}{\alpha}\right)
$$ 
examples where $x_i\sim D$ and $|\widehat{c}(x_i) - c(x_i)| < \zeta$ where $c\in \Cc$ is a unknown concept. There exists a $(T,\zeta^{-O(d)},O(\zeta))$-stable learning algorithm $\G$, 
that outputs $f$ satisfying $\Loss_{D}(f,c,O(\zeta)) \leq \alpha$.
\end{theorem}

The algorithm $\G$ is the $\RSOA$ run on a carefully tailored input distribution over the examples, with $T$ being the overall sample complexity of our algorithm. Most of the work in the proof arises in explaining how to tailor the set of examples drawn from the original distribution $D$ into a new set $S$ on which $\RSOA$ is guaranteed to succeed.  In this section, when we write $\RSOA_{\zeta}(S)$ where $S$ is a sample, i.e.,      $S = \{(x_i,\hat{c}(x_i))\}$, we mean that we feed the examples in $S$ into $\RSOA$ sequentially, as in the online learning setting. 
We will prove this theorem in three parts, corresponding to the subsequent three sub-subsections: 
\begin{itemize}
    \item 
    Our Algorithm~\ref{algo:distributions}, 
    is a tailoring algorithm that {\em defines}  distributions $\ext(\De,k)$ for $k\in[d]$ as a function of the distributions $\De$, to which we have black-box access. Just as in~\cite{bun2020equivalence}, the key idea for the tailoring is to inject examples into the sample that would force mistakes. We have adapted this idea for the robust, real-valued setting. Unfortunately, this algorithm could potentially use an unbounded number of examples (in the worst case), which we handle~next. 
    
    \item Next, we seek to impose a cutoff on the number of examples drawn in the algorithm above. In Lemma~\ref{lem:expectedsamplecomp}, we compute the expected number of examples drawn by Algorithm~\ref{algo:distributions}. Then, we use Markov's inequality to compute what the cutoff should be. The final tailoring algorithm is simply Algorithm~\ref{algo:distributions}, cut off when the number of examples drawn exceeds this threshold. 
    \item Finally, we state the globally-stable learning algorithm Algorithm~\ref{algo:stable_algo}, which essentially invokes Algorithm~\ref{algo:distributions} with the cutoff we defined above. In Theorem~\ref{thm:global_stability} we prove the correctness and sample complexity of Algorithm~\ref{algo:stable_algo}.
\end{itemize}

\subsubsection{Sampling from the distributions $\ext(D,k)$}
In the following, the symbol $S\circ T$ between two  sets of examples means the concatenation of the two sets $S,T$. Intuitively our learning algorithm is going to obtain $T$ examples overall and break these examples into blocks of size $m$ (a parameter which will be fixed later in Theorem~\ref{thm:global_stability}), each block followed by a single mistake example, all of which which are fed to an online learner. Additionally, below we can think of $k\leq \sfat(\Cc)$ as the number of mistakes we want to inject into the examples we feed to an online learner.

\begin{algorithm}
	\textbf{Input:} Distribution  $D:\mathcal{X}\rightarrow [0,1]$, $m\geq 1$, $k \in \{0,\ldots,d\}$.
	\vspace{2mm}

	\textbf{Output:} A sample from the distribution $\ext(D,k)$.\vspace{2mm}
	
For $k\geq 0$, the distributions $\ext(D,k): \X^{k(m+1)}\times [0,1]\rightarrow [0,1]$ are defined inductively as follows: 
	
\begin{enumerate}
    \item $\ext(D,0)$ : output the empty sample $\emptyset$ with probability 1.
     \item Sampling from $\ext(D,k)$ involves recursively sampling from $\ext(D,k-1)$ as follows:
    \begin{enumerate}[$(i)$]
        \item Draw $S^{(0)}, S^{(1)} \sim \ext(D,k-1)$ and two sets of $m$ examples $B^{(0)}, B^{(1)} \sim D^{m}$.
        \item $\operatorname{Let} f_{0}= \RSOA_{\zeta}\left(S^{(0)} \circ B^{(0)}\right), f_{1}=\RSOA_{\zeta}\left(S^{(1)} \circ B^{(1)}\right)$.
        \item If $|f_{0}(x)- f_{1}(x)| \leq 11\zeta $ for every $x \in \mathcal{X}$ then go back to step (i).
        \item Else pick $x'$ such that $|f_{0}(x')- f_{1}(x')| > 11 \zeta$ and sample $\alpha \sim \In_{\zeta}$ uniformly.\footnotemark
        \item Let $M_k := (x',\alpha)\in \X \times [0,1]$. If $|\alpha-f_0(x')|<|\alpha-f_1(x')|$, output $S^{(1)} \circ B^{(1)} \circ M_k$, else output $S^{(0)} \circ B^{(0)} \circ M_k$        
    \end{enumerate}
\end{enumerate}
\caption{An algorithm to sample from distributions $\ext(D,k)$.}
\label{algo:distributions}
\end{algorithm}
\footnotetext{Recall the definition of the $\zeta$-cover, $\In_{\zeta}=\big\{\zeta / 2,3 \zeta / 2, \ldots, 1 -\zeta/2 \big\}$}

\paragraph{Intuition of the algorithm.} We first explain Algorithm~\ref{algo:distributions} on an intuitive level. Recall the goal: using our $\RSOA$ online learning algorithm for $\Cc$, we would like to design a \emph{globally stable} $\PAC$ learner for $\Cc$. To this end, let $\De$ be the unknown distribution (under which we need the $\PAC$ learner to~work).    Algorithm~\ref{algo:distributions} `tailors' a sample (fed to the online learner) as follows: in the $k$th iteration it repeatedly draws pairs of batches of $(k-1)(m+1)$ examples from $\ext(D,k-1)$ and then decides whether to keep or discard each batch based on the outcome of running $\RSOA$ on the batches. If some batch is kept, it is appended with a {\em single} example which is guaranteed to force a mistake on $\RSOA$, and the resulting sample $S$ is output by the algorithm. This process of outputting $S$ can be regarded as drawing sample $S$ from the distribution $\ext(D,k)$. 
The structure of $S$ is illustrated in Figure~\ref{fig:curated_sample}. Each $B_i$ is a block of $m$ examples each drawn i.i.d.~from $D$. Each $M_i = (x_i, \alpha_i)$, forces a mistake when $S$ is fed to $\RSOA$. $S$ has $k$ blocks and $k$ mistake examples in total.

\begin{figure}[!ht]
    \centering
    \includegraphics[scale=1.2]{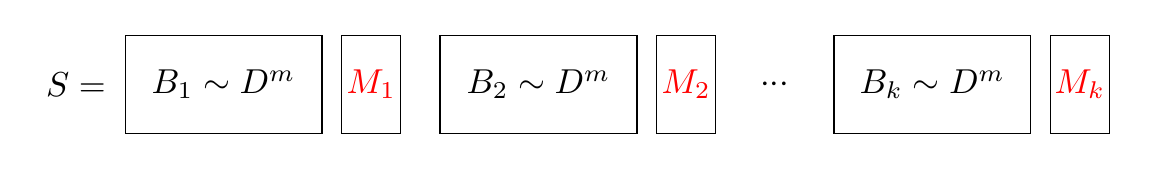}
    \caption{Structure of curated sample $S$ obtained  resulting from Algorithm~\ref{algo:distributions}. Each $B_i$ is a block of $m$ examples $(x,c(x))$ where $x\sim D$ and $M_i=(x,b)$ is an example which forces a \emph{mistake}.}
    \label{fig:curated_sample}
\end{figure}

We now focus on explaining steps 2$(i)$ to 2$(v)$ which `force a mistake'.  In step 2$(i)$ we draw two examples, $S^{(0)}\circ B^{(0)}$ and $S^{(1)}\circ B^{(1)}$. In 2$(ii)$, we feed $S^{(0)}\circ B^{(0)}$ into $\RSOA$, which returns function $f_0$, and do the same for $S^{(1)}\circ B^{(1)}$, returning $f_1$. There are now two possibilities, either $f_0,f_1$ are ``close" or $f_0$ and $f_1$ differ significantly at some $x\in \X$ and step 2$(iii)$ checks which is the case as follows.  
\begin{enumerate}
    \item $f_0,f_1$ agree to within $11\zeta$ on every point in $\X$: then draw a new pair $S^{(0)}\circ B^{(0)}$ and $S^{(1)}\circ B^{(1)}$ afresh, going back to step 2i). 
    \item $|f_0(x)-f_1(x)|>11\zeta$ for some $x \in \mathcal{X}$. Note that this $x$ need not be from an example previously given to the learner. Intuitively, in this case, the predictions $f_0$ and $f_1$ are so far apart at $x$ that they cannot both be $5\zeta$-correct, and so at least one of them is a mistake. More precisely, in the $\zeta$-cover, let $b_c\in \In_{\varepsilon}$ be the midpoint of the bin (of width $\zeta$) that contains $c(x)$. Since $|f_0(x)-f_1(x)|>11\zeta$, at least one of the predictions $f_0(x)$, $f_1(x)$ is $5\zeta$-far from $b_c$ (though we don't know which it is, since we don't know $c$!)
   \end{enumerate}
Steps 2$(i)$ to 2$(iii)$ are repeated until we are in the second case. Note that steps 2$(i)$ to 2$(iii)$ could be repeated an unbounded number of times, each repetition drawing fresh examples. For the remainder of this section, we assume that steps 2$(i)$ to 2$(iii)$ terminate eventually so that we may argue about the final output sample. In Section~\ref{sec:stoppingcriterion}, we show it suffices to ``impose" a cut-off of $T$ examples so that with high probability the algorithm (with an appropriate value of $k$) terminates before drawing $T$-many examples. 

In order to create $M_k$, we uniformly draw some $\alpha \sim \In_{\zeta}$ (the set of {all possible} bin midpoints), which means $\alpha = b_c$ with probability $\zeta$.\footnote{Note that this step crucially differs from~\cite{bun2020equivalence} since for them the true value of $f_0(x)$ or $f_1(x)$ is always $0$ or $1$, so they can flip a coin and force a mistake with probability at least $1/2$.}  If $\alpha = b_c$, we are guaranteed that $f_i$ is a mistake for $i := \arg \max_i |\alpha - f_i(x)|$. Therefore, we concatenate our mistake example with $S^{(i)} \circ B^{(i)}$, eventually outputting $S:=S^{(i)} \circ B^{(i)} \circ (x,\alpha)$ as the output of Algorithm~\ref{algo:distributions}. By the end of these steps, we will have a sample $S'\circ B' \circ M_k$ where $S' \sim \ext(D,k-1)$, $B' \sim D^m$ and $M_k$ is a single `mistake' example with the following two properties:
(i) $M_k = (x',\alpha)$ is a valid example (i.e.,      $|\alpha-c(x')| \leq \zeta$).
(ii) If $\RSOA$ is fed $S'\circ B' \circ M_k$, $\RSOA$ will make a mistake upon seeing the example $M_k$, i.e.,      at the round corresponding to $M_k$, $\RSOA$ predicts $\hat{y}$ such that $|\hat{y}-c(x')|>5\zeta$.

\paragraph{Key Lemma.} We now prove our key lemma on global stability. Let $\calR$ be the set of all possible functions that could be output by the $\RSOA$ algorithm when run for arbitrarily many rounds. \begin{lemma}[Some function ball is output by $\RSOA$ with high probability] \label{lem:freq}
 Let $\sfat_{2\zeta}{(\Cc)}=d$. There exists $k \leq d$ and some $f\in \mathcal{R}$ such that 
\begin{equation}
\Pr_{\substack{S \sim \ext(D,k),\\ B\sim D^{m}}}[\RSOA_{\zeta}(S\circ B)\in \T(5\zeta,f)] \geq \zeta ^{d}.
\end{equation}
\end{lemma}

\begin{proof}
Towards contradiction, suppose for every $k\leq d$ and $f \in \mathcal{R}$, we have
\begin{equation}\label{eq:supposition}
\Pr_{\substack{S \sim \ext(D,d),\\ B\sim D^{m}}}[\RSOA_{\zeta}(S\circ B) \in \T(5\zeta,f)] < \zeta^{d}.
\end{equation}
In particular, Eq.~\eqref{eq:supposition} holds for $f=c$ where $c$ is the target concept. 

In Step 2$(iv)$, Algorithm~\ref{algo:distributions} picks $\alpha$ uniformly from the set of midpoints in $\In_{\zeta}$. Call a mistake example $(x,\alpha)$ `valid' if $|\alpha-c(x)| \leq \zeta$. Notice there are actually two midpoints in $\In_{\zeta}$ which are less than $\zeta$ away from any $c(x)$, and hence, the probability that a mistake example is valid is $2\zeta > \zeta$. Hence the probability that all $d$ mistake examples are valid is at least $\zeta^d$. In the event that all mistake examples are valid, $S$ is a valid sample. Since $S$ contains $d$ mistake examples, and Theorem~\ref{thm:RSOA} guarantees that $\RSOA_{\zeta}$ on a valid sample always outputs some hypothesis function in $\T(5\zeta,c)$ after making $d$ mistakes, this contradicts Eq.~\eqref{eq:supposition}.
\end{proof}
\begin{lemma}[Generalization]\label{lem:generalizability_v2}
Let $\ext(D,\ell)$ be such that $\ell\geq 1$ and there exists $f$ such that
\begin{equation}\label{eq:LB}
\Pr_{\substack{S \sim \ext(D,\ell),\\ B \sim D^{m}}}[\RSOA_{\zeta}(S \circ B) \in \T(5\zeta,f)] \geq  \zeta^d.
\end{equation}
(The above property is the analog of the distribution $\ext(D,\ell)$ being `well-defined' in~\cite{bun2020equivalence}.)

Then, every $f$ satisfying Eq.~\eqref{eq:LB} also satisfies $\Loss_{D}(f,c,6\zeta) \leq d \ln(1/\zeta)/m$.
\end{lemma}

\begin{proof}
Let $S \sim \ext(D,\ell)$ and $B \sim D^{m}$. Suppose  $\RSOA_{\zeta}(S\circ B)$ outputs a function $f'\in \T(5\zeta,f)$. Now, for $f'\in \mathcal{R}$, let $E_{f'}$ be the event  that $\operatorname{\RSOA}_{\zeta}(S \circ B) $ outputs $f'$. Then observe that 
\begin{align}
\label{eq:sum_probs}
\begin{aligned}
\Pr_{\substack{S \sim \ext(D,\ell),\\ B\sim D^{m}}}[\RSOA_{\zeta}(S\circ B) \in \T(5\zeta,f)] &= \sum_{f':\, f' \in \T(5\zeta,f)}\Pr_{\substack{S \sim \ext(D,\ell),\\ B \sim D^{m}}}[E_{f'}]\\
&\leq \sum_{f':\, f' \in \T(5\zeta,f)} \Pr_{\substack{S \sim \ext(D,\ell),\\ B \sim D^{m}}}[B\text{ is }\zeta\text{-consistent with }f'] \\
&\leq \Pr_{\substack{S \sim \ext(D,\ell),\\ B \sim D^{m}}}[B\text{ is }6\zeta\text{-consistent with }f],
\end{aligned}
\end{align}
where the first inequality follows from combining two observations:
\begin{enumerate}
\item Since $B$ is a subset of the examples fed to $\RSOA_{\zeta}$, by Property 1 in  Lemma~\ref{lem:props_RSOA}, if $\operatorname{\RSOA}_{\zeta}(S \circ B) $ outputs $f'$ then $f'$ is $\zeta$-consistent with all $m$ examples in $B$;
\item By Property 4 of Lemma~\ref{lem:props_RSOA} (for a fixed sample, no two different functions can be output by $\RSOA$), $\{E_{f'}\}_{f' \in \mathcal{R}}$ are disjoint on the sample space;
\end{enumerate}
and the last inequality used  that $f'$ is in a $5\zeta$-ball of $f$, hence $f$ is $\zeta+5\zeta = 6\zeta$ consistent with $B$. 
Recall that Eq.~\eqref{eq:LB} shows that the LHS of Eq.~\eqref{eq:sum_probs} is lower-bounded by $\zeta^d$.  If we define  $\Loss_{D}(f,c,6\zeta) := \alpha$, then by the definition of loss, since $B$ is a sample of $m$ i.i.d.~examples drawn from $ D$, the RHS of the inequality above is $(1-\alpha)^m$.  Putting together the lower and upper bound  $\zeta^d \leq (1-\alpha)^m \leq e^{-\alpha m}$, proves  the lemma statement. 
\end{proof}

\subsubsection{A Monte Carlo version of the tailoring algorithm}\label{sec:stoppingcriterion}

Algorithm~\ref{algo:distributions} that we described in the previous section could potentially run steps $(i)-(iii)$ forever.  Apriori it is not clear why this algorithm terminates. In this section, 
we compute the expected number of examples drawn by Algorithm~\ref{algo:distributions} and eventually  use Markov's inequality to define a ``stopping criterion" (a sample complexity cutoff) on Algorithm~\ref{algo:distributions} so that the algorithm eventually stops drawing a certain number of examples. The reason the number of examples drawn is a random variable is that steps 2$(i)$ to 2$(iii)$ of Algorithm~\ref{algo:distributions} must be repeated until there is one round where~$f_0$, $f_1$ are distance more than $11\zeta$ apart, i.e.,      there exists  $x\in \mathcal{X}$ satisfying $|f_0(x) - f_1(x)| > 11\zeta$.

\begin{lemma}[Expected number of examples drawn in Steps 2$(i)$ to 2$(iii)$]\label{lem:expectedsamplecomp}
Let $\zeta \in [0,1/2]$ and let~$k^*$ be the smallest value (guaranteed to exist by Lemma~\ref{lem:freq}) for which 
\begin{align}
\label{eq:promiseoflem}
\Pr_{\substack{S \sim \ext(D,k^*),\\ B\sim D^{m}}}[\RSOA_{\zeta}(S\circ B) \in \T(11\zeta,f)] \geq \zeta^{d}
\end{align}
holds. Let $\ell \leq k^{\ast}$ and $M_{\ell}$ denote the number of examples drawn from $D$ in order to generate a sample $S \sim \ext(D,\ell)$. Then
\[
\mathbb{E}\left[M_{\ell}\right] \leq 4^{\ell+1} \cdot m,
\]
where the expectation is taken over the random sampling process in Algorithm~\ref{algo:distributions}.
\end{lemma}

\begin{proof}
Because we have chosen $k^{\ast}$ to be the smallest value for which Eq.~\eqref{eq:promiseoflem} is true, this implies that for every $\ell'<k^{\ast}$ and $ f \in \mathcal{R}$, we have
$$ 
\Pr_{\substack{S \sim \ext(D,\ell'),\\ B \sim D^{m}}}[\RSOA_{\zeta}(S\circ B) \in \T(11\zeta,f)] < \zeta^d 
$$
  which is equivalent to  
    $$
    \Pr_{\substack{S \sim \ext(D,\ell'),\\ B \sim D^{m}}}[\RSOA_{\zeta}(S\circ B) \notin \T(11\zeta,f)] \geq 1- \zeta^d.\label{eq:unconcentrated}
$$
Now consider sampling from $\ext(D,\ell)$ such that $0\leq \ell \leq k^{\ast}$. Call each round of 2$(i)$ to 2$(iii)$ `successful' if it results in $f_0$, $f_1$ such that $|f_0(x)-f_1(x)|> 11\zeta$ for some $x$. 
Upon success, the algorithm proceeds to step 2$(iv)$. Let us assume that the probability of success for the $\ell$th round is  $\theta$. Then one can express $\theta$ as follows: 
\[
\begin{aligned}
\theta &=\sum_{f_0 \in \mathcal{R}} \Pr_{\substack{S_0\sim \ext(D,\ell-1),\\ B_0\sim D^{m}}}[\RSOA(S_0 \circ B_0)=f_0]\cdot   \Pr_{\substack{S_1\sim \ext(D,\ell-1),\\ B_1\sim D^{m}}}[\RSOA(S_1 \circ B_1 ) = f_1, \,f_1\not \in \T(11\zeta,f_0)]\\
& \geq (1-\zeta^d) \sum_{f_0 \in \mathcal{R}} \Pr_{\substack{S_0\sim \ext(D,\ell-1),\\ B_0\sim D^{m}}}[\RSOA(S_0 \circ B_0)=f_0] =1-\zeta^d,
\end{aligned}
\]
where the first equality is because `success' is defined as $|f_0(x)-f_1(x)|> 11\zeta$ at some $x$, equivalently $f_1\not \in \T(11\zeta,f_0)$, and we used Eq.~\eqref{eq:unconcentrated} in the inequality.

Furthermore, sampling from $\ext(D,\ell)$ involves sampling from $\ext(D,\ell-1),\ldots, \ext(D,0)$. Therefore, the number of examples drawn to sample from $\ext(D,\ell)$, $M_{\ell}$, is a function of $M_{\ell-1},\ldots ,M_0$. Let $M_{\ell}^{(j)}$ be the number of examples drawn during the $j$th attempt at sampling from distribution $\ext(D,\ell)$ and write
$    M_{\ell}=\sum_{j=1}^{\infty} M_{\ell}^{(j)}.
$
While sampling from distribution $\ext(D,\ell)$, if we succeed prior to the $j$-th attempt, $M_{\ell}^{(j)}=0$; otherwise, if the first $j-1$ attempts end in failure, we have to draw two examples from $\ext(D,\ell-1)$ and two examples from $D^m$. Therefore, we may define the recursive equation
\begin{equation}
\mathbb{E}\left[M_{\ell}^{(j)}\right]=(1-\theta)^{j-1} \cdot\left(2 \mathbb{E}\left[M_{\ell-1}\right]+2 m\right),
\end{equation}
since each attempt involves drawing two examples from $\ext(D,\ell-1)$ and two examples from $D^{m}$ and we used the fact that the probability of failure is $(1-\theta)^{j-1}$.
Therefore, we have
\begin{align}\label{eq:boundonE(M_l)}
\begin{aligned}
\mathbb{E}\left[M_{\ell}\right]=\sum_j\mathbb{E}\left[M_{\ell}^{(j)}\right]&=\sum_{j=1}^{\infty}(1-\theta)^{j-1} \cdot\left(2 \mathbb{E}\left[M_{\ell-1}\right]+2 m\right)\\
&=\frac{1}{\theta} \cdot\left(2 \mathbb{E}\left[M_{\ell-1}\right]+2 m\right) \\
&\leq \frac{1}{1-\zeta^d}\cdot ( 2\mathbb{E}\left[M_{\ell-1}\right]+2 m)\leq 4\cdot ( \mathbb{E}\left[M_{\ell-1}\right]+m),
\end{aligned}
\end{align}
where we have used the fact that $\zeta<1/2$ to obtain the last inequality. Using that $\mathbb{E}[M_0]=0$ and using induction on Eq.~\eqref{eq:boundonE(M_l)} gives us the lemma statement. 
\end{proof}

\subsubsection{Final algorithm}
Putting together these pieces, we now prove our main theorem.

\begin{theorem}[Globally stable learner from online learner] \label{thm:global_stability}
Let $\alpha>0$. Let $\Cc\subseteq \{f:\X\rightarrow [0,1]\}$ be a concept class with $\sfat_{2\zeta}(\Cc) = d$. Let $c\in \Cc$ be the target concept. Let
\[
T =\left(2\cdot (4/\zeta)^{d+1}+1\right) \cdot \frac{d \ln (1/\zeta)}{\alpha}.
\]
Let $D:\X\rightarrow [0,1]$ be a distribution. There exists a randomized algorithm $G:(\mathcal{X} \times[0,1])^{T} \rightarrow [0,1]^{\mathcal{X}}$ that satisfies the following: given $T$ many examples $S=\{(x_i,\widehat{c}(x_i))\}$ where $x\sim D$, there exists a hypothesis $f$ such that
\begin{equation}\label{eq:relaxed_stability}
\operatorname{Pr}[G(S) \in \T(11\zeta,f)] \geq \frac{\zeta^{d}}{2(d+1)} \text { and } \Loss_{D}(f,c, 12\zeta) \leq \alpha
\end{equation}
\end{theorem}

\begin{proof}
The algorithm $G$ in the theorem statement is exactly the algorithm we defined in the previous two sections along with a cutoff at $T$ examples.

\begin{algorithm}
\begin{enumerate}
\item Draw $k \in\{0,1, \ldots, d\}$ uniformly at random.
\item Let $\ext(D,k)$ be the distribution described in  Algorithm~\ref{algo:distributions} but additionally imposing a cutoff $T$ on sample complexity (i.e., we output `fail' if the number of examples drawn in sampling from $\ext(D,k)$ ever exceeds $T$), where the auxiliary sample size is set to $m= d\ln (1/\zeta) /\alpha$ and cutoff $T= 2\cdot (4/\zeta)^{d+1} \cdot m$.\footnotemark 

Let $B \sim D^{m}$ and $S \sim \ext(D,k)$ and output $h=\RSOA_{\zeta}(S \circ B)$.
\end{enumerate}
\caption{Final globally-stable algorithm $G$ to learn concept class $\Cc\subseteq \{f:\mathcal{X}\rightarrow [0,1]\}$.
}
\label{algo:stable_algo}
\end{algorithm}
Note that because we have enforced the cutoff at $T$ examples in drawing $S \sim \ext(D,k)$, the sample complexity of $G$ is $|S|+|B| \leq T+m=\left(2\cdot (4/\zeta)^{d+1}+1\right) \cdot \frac{d \ln (1/\zeta) }{\alpha}$ as stated in the theorem statement.  Lemma~\ref{lem:freq} guarantees that there exists $k \leq d$ and $f^{\ast}$ such that Eq.~\eqref{eq:LB} holds. Let $k^{\ast}$ be the smallest $k$ such that Lemma~\ref{lem:freq} holds with the constant $5\zeta$ replaced by $11\zeta$, and
\begin{equation}
\label{eq:promiseoflem_2}
\Pr_{\substack{S \sim \ext(D,k^\ast),\\ B\sim D^{m}}}[\RSOA_{\zeta}(S\circ B) \in \T(11\zeta,f^{\ast})] \geq \zeta^{d}.
\end{equation}
Then Lemma~\ref{lem:generalizability_v2} (with a simple modification for the new constant) implies that
%\begin{equation}
  $  \Loss_{D}(f,c,12\zeta) \leq  d\ln (1/\zeta) /m  \leq \alpha$.
%\end{equation}
\footnotetext{For simplicity in notation, we assume $cd/\alpha$ is an integer. If not, one can set $m=\lceil cd/\alpha \rceil$.} 

We now show that the probability that $G$ outputs some function in $\T(11\zeta,f^{\ast})$ is $\frac{1}{2(d+1)}\cdot \zeta^{d}$. Firstly, with probability $\frac{1}{d+1}$, the randomly drawn $k$ in step 2 is $k^{\ast}$. Conditioned on this, we now show that with high probability, the loop in Steps 2$(i)$ to 2$(iii)$ will terminate after drawing $T = 2\cdot (4/\zeta)^{d+1} \cdot m$ examples. 
\begin{align}
\label{eq:stoppingcriteron}
    \Pr\big[M_{k^{\ast}} > 2\cdot (4/\zeta)^{d+1}\cdot m\big] &\leq \Pr\big[M_{k^{\ast}} > 2\cdot \zeta^{-d}\cdot 4^{k^{\ast}+1}\cdot m\big]
    \leq \zeta^{d}/2,
\end{align}
where the first inequality used $k^*\leq d$ and
 the second inequality is by Markov's inequality and Lemma~\ref{lem:expectedsamplecomp}.  Putting together Eq.~\eqref{eq:promiseoflem_2} and~\eqref{eq:stoppingcriteron} 
 the probability that $\RSOA(S\circ B)$ outputs a function in $\T(11\zeta,f^*)$
 %$\T(11\zeta,\RSOA(S\circ B) \in \T(11\zeta,f))$
 and also Algorithm~\ref{algo:distributions} terminates before the cutoff $T$ is
\begin{equation}
\begin{aligned}
&\Pr_{\substack{S \sim \ext(D,k^\ast),\\ B\sim D^{m}}}\left[\RSOA(S\circ B) \in \T(11\zeta,f^{\ast})\text { and } M_{k^{*}} \leq 2\cdot (4/\zeta)^{d+1} \cdot m\right] 
 \geq \zeta^{d} - \zeta^{d}/2 = \zeta^{d}/2
\end{aligned}
\end{equation}
Multiplying this together with $1/(d+1)$ yields our claim.
\end{proof}

\subsection{Stability does not imply approximate $\DP$ (without a domain-size dependence)} \label{sec:stabdoesntimplyapproxDP}

In the previous section we showed that if a concept class $\Cc$ can be learned in the quantum online learning framework, then there exists a globally stable learner (with appropriate parameters) for $\Cc$ as well. This implication was first pointed out by \cite{bun2020equivalence} for Boolean-valued $\Cc$s. In fact, they went one step further and created a \emph{approximately differentially-private} learner from a stable learner. In this sense, stability can be viewed as an intermediate property between online learnability and approximate differential privacy in the Boolean setting. Jung et al.~\cite{jung2020equivalence} used the same technique to show that stability implies approximate differential privacy in the {\em multiclass} learning setting as well (i.e., when the concept class to be learned maps to a discrete set $\{1,\ldots,k\}$), but they do not show that an analogous implication holds for real-valued learning, which they mention briefly. Note that their real-valued learning setting is less general than ours, as they assume that they receive exact feedback on each example (we discuss this at the end of this section). 

A natural question is: does this result still hold in the quantum learning setting, i.e., does quantum stability imply quantum differential privacy? In this section, we show that the \cite{bun2020equivalence} method for showing this implication for Boolean functions -- which held up in the case of learning multiclass functions -- fails for learning real-valued functions with imprecise feedback. Unlike in the former two cases, the transformation from stable learner to approximate $\DP$ learner necessarily incurs a domain-size dependence in the sample complexity. This is undesirable because, when $\X$ is a real-interval or if it is unbounded, this quantity could potentially be infinite. 

\subsubsection{Sample complexity of stability to privacy transformation}\label{subsubsec:stabletoprivate}

In the Boolean setting,~\cite{bun2020equivalence} showed that one could use the stable histograms algorithm~\cite{bun2019simultaneous} and the Generic Private Learner of~\cite{kasiviswanathan2011can}, to convert a Boolean globally-stable learner, in a black-box fashion, to a private learner. This learner's sample complexity depends on $\Ldim(\Cc)$ and the privacy and accuracy parameters of the stable learner, but \emph{not} the domain size of the function class. We now show that this technique cannot possibly yield a domain size-independent sample complexity for quantum learning.

Our stable learner $G$ has the following guarantees (given in Theorem~\ref{thm:global_stability}):  there exists some {\em function ball} (around the target concept) such that the collective probability of $G$ outputting its member functions is high. Contrast this with the global stability guarantee for learning Boolean functions~\cite{bun2020equivalence}, which says that $G$ outputs some {\em fixed function} with high~probability. The stability guarantees differ because, in our setting, the learner only obtains $\eps$-accurate feedback from the adversary. Hence the learner cannot uniquely identify the target concept $c$, since all functions that are in the $\varepsilon$-ball of $c$ would be consistent with the feedback of the adversary, and we thus allow the learner to output a function in the $\varepsilon$-ball around the target concept. However, this difference critically prevents us from using the~\cite{bun2020equivalence} technique to transform a stable learner into a private learner in the quantum case. We sketch this argument below\footnote{The following argument was communicated to us by Mark Bun~\cite{bununpublished}.}, which relies on ideas from classical fingerprinting codes~\cite{bun2018fingerprinting} (which were also used earlier by Aaronson and Rothblum~\cite{aaronson2019gentle} in order to give lower bounds on gentle shadow tomography).

\cite{bun2020equivalence}'s transformation from stable learner to private learner, applied to our setting, would be as follows: generate a list of functions in $\Cc$ by running the stable learner $G(S)$ of Theorem \ref{thm:global_stability},~$n$ many times, each of which outputs a single $f_i \in \Cc$. By Theorem \ref{thm:global_stability} and a Chernoff bound, one can show that with high probability, an $\eta = \zeta^d$-fraction of the list should be in $\T(\zeta, f^{\ast})$ for some~$f^{\ast}$. Next one would like to \emph{privately} output some function in $\T(\zeta, f^{\ast})$. We rewrite this now as follows.

\begin{problem}[Query release for function balls]
\label{prob:stabletoprivate}
Given a list of $n$ functions $\{f_i: \mathcal{X} \rightarrow \mathbb{R}\}_{i\in [n]}$, an $\eta$-fraction of which are in $\T(\zeta, f^{\ast})$ for some $f^{\ast}: \mathcal{X} \rightarrow \mathbb{R}$,
output some function $g \in \T(\zeta, f^{\ast})$.
\end{problem}

We could also consider the following problem of clique identification on a discrete domain.

\begin{problem}[Clique identification on a discrete domain]
\label{prob:clique}
{\em Clique identification} is the following problem: given a symmetric, reflexive relation $R \subseteq \mathcal{Y} \times \mathcal{Y}$ and a dataset $D \in \mathcal{Y}^{n}$ under the promise that $(x, y) \in R$ for every $x, y \in D,$ find any point $z \in \mathcal{Y}$ such that $(x, z) \in R$ for every $x \in D$. {\em Clique identification on a discrete domain} is {\em clique identification} with $\mathcal{Y} = [4]^d$ and $R = \{(x,y) \in \mathcal{Y}\times \mathcal{Y} :\| x-y \|_{\infty} \leq 1 \}$. 
\end{problem} 

Problem \ref{prob:clique} is a special case of Problem \ref{prob:stabletoprivate}, when we choose the functions $f$ to be of the form $f:[d] \rightarrow [4]$, $\eta=1$ and $\zeta = 1/2$, and let $D$ consist of the $n$ vectors $[f_i(1), \ldots f_i(d)], \, i\in [n]$. Hence, any $\DP$ algorithm for query release for function balls is also a $\DP$ algorithm for clique identification on a discrete domain. However, we claim the following: 

\begin{claim}\label{claim}
For $\delta< 1/1500$, any $(1,\delta)$-$\DP$ algorithm\footnote{It is not hard to modify this proof so as to allow an $\eps$ privacy parameter.} solving Problem \ref{prob:clique} with probability at least $1499/1500$ requires $n \geq \tilde{\Omega}(\sqrt{d})$.
\end{claim}

We will prove the claim later, but we first explain why it implies a necessary domain size dependence in the transformation we hope to achieve. Noting that $d = |\mathcal{X}|$ in the translation from Problem \ref{prob:stabletoprivate} to Problem \ref{prob:clique},  we conclude from Claim \ref{claim} that any $(1,\delta)$-$\DP$ algorithm for Problem \ref{prob:stabletoprivate} requires $n \geq \tilde{\Omega}(\sqrt{|\X |})$. Hence, any algorithm to convert the stable real-valued learner $G$ of Theorem \ref{thm:global_stability} into an approximate-$\DP$ learner that also solves Problem \ref{prob:stabletoprivate}, also requires to run the stable learner $n$-many times, each of which consumes $T$ examples. Hence the total number of examples needed is 
\begin{equation}
\label{eq:lowerboundreal}
   \tilde{\Omega}\left(\sqrt{|\mathcal{X}|}  \left(2\cdot (4/\zeta)^{d+1}+1\right) \cdot \frac{d \ln (1/\zeta)}{\alpha} \right).
\end{equation}

In particular, this lower bound is also optimal for query release up to poly-logarithmic factors, i.e., using $\tilde{O}(\sqrt{ |\X|})$ examples one can solve Problem~\ref{prob:stabletoprivate} using the Private Multiplicative Weights method by Hardt and Rothblum~\cite{hardt2010multiplicative} (as also referenced in the work of Bun et al.~\cite{bun2018fingerprinting}).

To prove Claim \ref{claim}, we first need to first define weakly-robust fingerprinting codes (first introduced by Boneh and Shaw~\cite{boneh1998collusion}, then developed in~\cite{bun2018fingerprinting}). 
 
 \begin{definition}
 An $(n,d)$-fingerprinting code with security $s$ and robustness $r$ is a pair of random
variables $(G,T)$ where $G\in \{2,3\}^{n\times d}$ and $T:\{2,3\}^{d}\rightarrow2^{[n]}$ that satisfy the following. We say that a column $j\in [d]$ is marked if there exists $b\in \{2,3\}$ such that $x_{i;j} = b$ for all $i\in [n]$. Similarly, we say a string $w\in\{2,3\}^d$ is feasible for $G$ if for at least a $1-r$ fraction of the marked columns $j\in G$, the entry $w_j$ agrees with the common value in that column. The code must satisfy the properties of soundness and completness, as follows:

\noindent \emph{Completeness.} For every $A:\{2,3\}^{n\times d} \rightarrow \{2,3\}^d$, 
$\Pr_{w\leftarrow A(G)} [w \text{ is feasible for } G \text{ and } T(w)=\emptyset]\leq s$

\noindent \emph{Soundness.} For every $i\in [n]$, algorithm $A:\{2,3\}^{n\times d} \rightarrow \{2,3\}^d$, we have $\Pr_{w\leftarrow A(G_{-i})}[T(w)\ni i]\leq~s$
\end{definition}

\noindent We also need the following result for explicit construction of fingerprinting codes.

\begin{theorem}[\label{thm:tardos}\cite{tardos2008optimal}]
Then, for every $s \in (0,1)$, there exists an $(n,d)$-fingerprinting code with security $s$ and robustness $r=1/25$ with $d=\tilde{O}(n^2\log (1/s))$.
 \end{theorem}
 
 With this we now prove our main claim.
 
 \begin{proof}[Proof of Claim \ref{claim}]
The idea is to construct, from any $(\eps=1, \delta = 1/4n)$-$\DP$ clique identification algorithm with success probability at least $1499/1500$, an adversary $A:\{2,3\}^{n\times d} \rightarrow \{2,3\}^d$ for any $(n,d)$ fingerprinting code with robustness $1/25$, such that the code cannot be $1/20n$-secure against the adversary. However, because Theorem \ref{thm:tardos} guarantees the existence of a sound and complete $(n,d)$-fingerprinting code with $(s=1/20n,r=1/25)$-parameters as long as $n < \tilde{\Omega}(\sqrt{d})$, the claimed clique identification algorithm $M$ must have $n \geq \tilde{\Omega}(\sqrt{d})$. We now go into more detail about how to construct the adversary. 

Let $M$ be the alleged $\DP$ algorithm for clique identification, and let $G\in \{2,3\}^{n\times d}$ be the $G$ corresponding to the fingerprinting code. If we regard each of the rows of $G$ as being a point in $\mathcal{Y} = [4]^d$, then taking $D$ to be the set of all rows of $G$, $D$ fulfils the promise of Problem \ref{prob:clique}. Then the adversary $A$ is constructed out of $M$ as follows: on input $D$, run $M(D)$ producing a string $w\in [4]^d$. Return the string $w^{\prime} \in \{2,3\}^d$ where $w_{i}^{\prime}=2$ if $w_{i} \in\{1,2\}$ and $w_{i}^{\prime}=3$ if $w_{i} \in\{3,4\}$. A proof by contradiction, which we omit, shows that the string $w'$ produced in this manner is feasible for the fingerprinting code with probability at least $2/3$. By completeness of the code, $\operatorname{Pr}[T(A(D)) \in[n]] \geq 2 / 3-s \geq 1 / 2 .$ In particular, there exists some $i^{*} \in[n]$ such that $\operatorname{Pr}\left[T(A(D))=i^{*}\right] \geq 1 / 2 n .$ Now by differential privacy,
$$
\operatorname{Pr}\left[T\left(A\left(D_{-i^{*}}\right)\right)=i^{*}\right]  \geq e^{-\varepsilon}\left(\operatorname{Pr}\left[T(A(D))=i^{*}\right]-\delta\right) \geq e^{-1}\left(\frac{1}{2 n}-\frac{1}{4 n}\right)  \geq \frac{1}{20 n} .
$$
This contradicts the soundness of the code.
\end{proof}

\subsubsection{A quadratically worse upper bound on the sample complexity of privacy}
The previous section showed that going from a stable learner to a private learner of real-valued function classes should incur a sample complexity at least the square root of domain size. We now show to obtain a pure-$\DP$ learning algorithm for real-valued function classes over a finite domain (with no need for the stability intermediate step) that needs at most linear-in-$|\mathcal{X}|$ examples, which is  quadratically worse than the lower bound. This was also pointed out in the Appendix of~\cite{jung2020equivalence}. 

The private algorithm that accomplishes this is the Generic Private Learner of~\cite{kasiviswanathan2011can,bun2020equivalence}. We give its guarantees in the lemma below. Intuitively, this lemma states that given a collection of hypotheses, one of which is guaranteed to have low loss $\alpha$ with respect to some unknown distribution and target concept, by adding Laplace noise, one can \emph{privately} output with high probability a hypothesis with loss at most $2\alpha$ with respect to the unknown target concept and~distribution. 

\begin{lemma}[Generic Private Learner~\cite{kasiviswanathan2011can,bun2020equivalence}\label{lem:GL}]
Let $\Hi \subseteq \{h:\X\rightarrow [0,1]\}$ be a set of hypotheses. For
$$
m=O\left(\frac{\log |\Hi|}{\alpha \varepsilon}\right)
$$
there exists an $(\varepsilon,0)$-differentially private generic learner $\GL: (\X \times[0,1])^{m} \rightarrow \Hi$ such that the following holds. Let $D: \X \times [0,1] \rightarrow [0,1]$ be a distribution, $c:\X\rightarrow [0,1]$ be a target function, $\zeta$ be a distance parameter and  $h^{*} \in \Hi$ be such that with
$
\Loss_{D}\left(h^{*},c,\zeta \right) \leq \alpha.
$
Then on input $S \sim D^{m}$, algorithm $\GL$ outputs, with probability at least $2/3$, a hypothesis $\hat{h} \in \Hi$ such that
$\Loss_{D}(\hat{h},c,\zeta) \leq 2 \alpha.$
\end{lemma}

For every real-valued function class $\Cc$, one could discretize the $[0,1]$-range of its functions $h:\mathcal{X}\rightarrow [0,1]$ into bins of size $\zeta$. This obtains a discretized function class $\mathcal{H}$ with at most $(1/\zeta)^{|\mathcal{X}|}$ functions. Plugging this bound into the lemma above, we obtain a private learner with sample~complexity 
\begin{equation}\label{eq:upperboundreal}
m=O\left(\frac{|\mathcal{X}| \log (1/\zeta)}{\alpha \varepsilon}\right). 
\end{equation}

\subsection{Implications for quantum learning}
\label{sec:quantumimplications}
We now turn to the quantum implications of the results in the previous sections. While we have stated all our results for the case of learning real-valued functions with imprecise adversarial feedback, we now expressly translate them to the setting of learning quantum states. Recall that, as stated in Section \ref{sec:prelim}, in quantum learning we are given $\U$, a class of $n$-qubit quantum states from which the state to be learned is drawn; $\M$, a set of $2$-outcome measurements  and $D:\M\rightarrow [0,1]$, a distribution on the set of measurements.\footnote{To be more clear, $D$ can be viewed as a distribution over $\{(E_i,\id-E_i)\}_i$ where $\{E_i\}_i$ is an orthogonal basis for the space of operators on $n$-qubits satisfying $\|E_i\|\leq 1$.} Our results apply to quantum learning by associating, to every $\rho\in \U$, the real-valued function $c_\rho:\M\rightarrow [0,1]$ defined as $c_\rho(M)=\Tr(M\rho) \in [0,1]$ for every $M\in \X$, and taking the function class to be $\Cc_{\U} = \{c_{\rho}\}_{\rho \in \U}$.  Section \ref{sec:onlineimpliesstab} implies that given a $\Cc_{\U}$ with bounded $\sfat$ dimension, a stable learner for $\Cc_{\U}$ also exists. To translate this result into the quantum learning setting, we define quantum stability as follows:
\begin{definition}[Quantum stability]
A quantum learning algorithm~$\calA: (\M \times [0,1])^T \rightarrow \U$ is $(T,\varepsilon,\eta)$-\emph{stable} with respect to distribution $D:\M\rightarrow [0,1]$ if, given $T$ many labelled examples $S=\{(E_i,y_i)\}_{i\in [T]}$ where $|\Tr(\rho E_i) - y_i| <\zeta$, there exists a state $\sigma$ such that 
\begin{equation}\label{eq:stable_maintext}
\Pr[\calA(S) \in \calB_{\M}(\varepsilon, \sigma)]\geq \eta,
\end{equation}
where the probability is taken over the examples in $S$ and $\calB_{\M}(\varepsilon,\sigma) := \{\rho: |\Tr(E\rho)-\Tr(E\sigma)| \leq \varepsilon\}$, that is to say, the ball of states within distance $\eps$ of $\sigma$ on $\M$.
\end{definition}
In other words, quantum stability means that up to an $\eps$-distance on the measurements in $\M$, there is some $\sigma$ that is output by $\calA$ with ``high" (at least $\eta$) probability. Then the quantum version of Theorem \ref{thm:global_stability} is the following:

\begin{theorem}[Quantum-stable learner from online learner]
\label{thm:quantum_global_stability}
Let $\U$ be a class of quantum states with $\sfat_{2\zeta}(\Cc_{\U}) = d$, let $\M$ be a set of orthonormal $2$-outcome measurements and let $D:\M\rightarrow [0,1]$ be a distribution over measurements. There exists an algorithm~$\G: (\M \times [0,1])^T \rightarrow \U$ that satisfies the following: for every $\rho\in \U$, given 
\[T =\left(2\cdot (4/\zeta)^{d+1}+1\right) \cdot \frac{d \ln (1/\zeta)}{\alpha}.\]
many labelled examples $S=\{(E_i,y_i)\}_{i\in [T]}$ where $|\Tr(\rho E_i) - y_i| <\zeta$ and $E_i \sim D$, there exists a $\sigma$ such that $\Pr_{S\sim D^T}[\G(S) \in \B_{\M}(11\zeta,\sigma)] \geq \frac{\zeta^d}{2(d+1)}$ and $\Pr_{E\sim D}\big[|\Tr(\rho E)-\Tr(\sigma E)|\leq 12 \zeta\big] \geq 1-\alpha$. 
\end{theorem}
Namely, $\G$ is $(T,11\zeta, \frac{\zeta^d}{2(d+1)})$-stable and furthermore, the state $\sigma$ has loss $\alpha$.  Section \ref{subsubsec:stabletoprivate} now gives a no-go result for going from the above-mentioned quantum-stable learner to an approximate-$\DP$ one. It shows that the technique of \cite{bun2020equivalence} to convert a stable learner to a private one necessarily incurs a domain-size dependence in the sample complexity. 

We say a few words about the implications of this on quantum learning. As explained earlier, it is often of most interest to choose $\M$ to be some orthogonal set of measurements. If, say, we choose it to be the orthogonal basis of $n$-qubit Paulis, then $|\M| = 4^n$ and so Equation \eqref{eq:lowerboundreal} implies that one needs sample complexity $\tilde{\Omega}(4^{n/2})$ in order to go from stability to approximate differential~privacy, whereas Equation \eqref{eq:upperboundreal} implies that even without stability, there exists a simple (pure) private learner for $\Cc_{\U}$ whose sample complexity is $\tilde{O}(4^{n})$, which is quadratically worse. 

\section{Pure differential privacy implies online learnability}
\label{sec:pureDPimpliesonline}
In this section we will prove the converse direction of the implication we showed in the previous section, namely that $\DP$ $\PAC$ learnability of a concept class $\Cc$ implies online learnability of $\Cc$.  To be more precise, we will show that the sample complexity of {\em pure} $\DP$ $\PAC$ learning $\Cc$ is linearly related to the $\sfat(\cdot)$ dimension of $\Cc$. 
Combining this with Theorem~\ref{thm:RSOA} implies learnability in the pure $\DP$ $\PAC$ setting implies online learnability of $\Cc$ in the strong feedback setting. The implications we will show are summarized in the diagram below:
\begin{figure}[!ht]
\centering
\begin{tikzpicture}
[->,>=stealth',shorten >=1pt,auto,  thick,yscale=0.8,
main node/.style={circle,draw}, node distance = 0.8cm and 1.8cm,
block/.style   ={rectangle, draw, text width=5em, text centered, rounded corners, minimum height=2.5em, fill=white, align=center, font={\footnotesize}, inner sep=5pt},
smallblock/.style   ={rectangle, draw, text width=3em, text centered, rounded corners, minimum height=2.5em, fill=white, align=center, font={\footnotesize}, inner sep=5pt},
medblock/.style   ={rectangle, draw, text width=4em, text centered, rounded corners, minimum height=2.5em, fill=white, align=center, font={\footnotesize}, inner sep=5pt}]
        \node[main node,smallblock] (PureDPPAC) at (-9,2) {\textsf{Pure\\ DP PAC}};
    \node[main node,block] (Prdim) at (-5.5,2) {\textsf{Representation dimension}};
    \node[main node,medblock] (Roneway) at (-1.6,2) {\textsf{One-way CC}};
   \node[main node,block] (sfat) at (2.1,2) {\textsf{Sequential fat-shattering dimension}};
  \node[main node,medblock] (online) at (6,2) {\textsf{Online learning}};
  
    \path [->] (PureDPPAC) edge node {Lem~\ref{lem:SCDPtoPRdim}} (Prdim);
    \path [->](Prdim) edge node {Lem~\ref{lem:PRdimtoR}} (Roneway);
    \path [->](Roneway) edge node {Lem~\ref{lem:R(C)tosfat(C)}} (sfat);
    \path [->](sfat) edge node {Thm~\ref{thm:RSOA}} (online);
\end{tikzpicture}
    \caption{Sample complexity of pure $\DP$ $\PAC$ upper-bounds $\sfat(\cdot)$.}\label{fig:DPtosfat}
\end{figure}
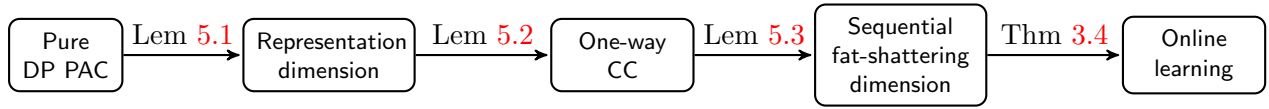

This section is organized as follows. 
 In Section~\ref{sec:puredpimpliesoneway} we show that the sample complexity of pure $\DP$ $\PAC$ is linearly related to the communication complexity of one-way public communication. As shown in Figure~\ref{fig:DPtosfat}, the link between these two notions goes through representation dimension. 
 In Section~\ref{sec:onewayandsfat} we show that one-way communication complexity is, in turn, characterized by $\sfat(\cdot)$. Additionally, we know from Theorem~\ref{thm:RSOA} that this combinatorial dimension upper-bounds the mistake bound of online learning $\Cc$, and this completes the chain of implications shown in Figure~\ref{fig:DPtosfat}. 
%\end{itemize}

\subsection{Pure $\DP$ $\PAC$ implies one-way communication}
\label{sec:puredpimpliesoneway}

In this section we prove that the sample complexity of pure $\DP$ $\PAC$ learning upper bounds one-way communication complexity of a concept class $\Cc$. 

\subsubsection{Pure differential privacy and $\PRD$}
We start by relating the sample complexity of differentially-private $\PAC$ ($\PPAC$) learning (see Definition~\ref{def:PRPAC}) a concept class $\Cc$, to the probabilistic representation dimension of $\Cc$. As in the previous section, we use the shorthand $S\sim D^m$ to mean that the sample $S$ is of the form $\{(x_i,\widehat{c}(x_i))\}_{i=1}^m$ where each $x_i \sim D$ and for all $i$, $\widehat{c}(x_i)$ satisfies $|\widehat{c}(x_i)- c(x_i)| < \zeta/5$. 

\begin{lemma}[Sample complexity of $(\zeta, \alpha, \eps,0)$-$\PPAC$ learning and $\PRD$]
\label{lem:SCDPtoPRdim}
Let $\alpha<1/4.$ Suppose there exists an  algorithm $\A$ that $(\zeta, \alpha, \eps,0)$-$\PPAC$ learns a real-valued concept class $\Cc\subseteq \{f:\X\rightarrow [0,1]\}$ with sample size~$m$, then there exists a set of concept classes $\mathscr{H}$ and a distribution over their indices $\mathcal{P}$, such that  $(\mathscr{H},\mathcal{P})$  $(\zeta,1/4,1/4)$-probabilistically represents  $\Cc$, with $\operatorname{size}(\mathscr{H})=O(m \eps \alpha)$. This implies that the sample complexity of $(\zeta, \alpha, \eps,0)$-$\PPAC$ learning $\Cc$ is 
\begin{equation}
    \Omega \left(\frac{1}{\alpha \eps} \PRD_{\zeta,1/4,1/4}(\Cc) \right) .
\end{equation}
\end{lemma}

\begin{proof}
Our proof extends the work of Beimel et al.~\cite{beimel2013} to the case of robust real-valued $\PAC$ learning. We assume we are given a $(\zeta,\alpha,\eps,0)$-$\PPAC$ learner $\A$ of $\Cc$ that outputs some function in hypothesis class $\mathcal{F}$ with sample complexity $m$. The $\PAC$ guarantees hold whenever the feedback is a $\zeta/5$ approximation of $c(x_i)$, so for the rest of this proof, we will fix the examples $(x_i, \widehat{c}(x_i))$ to have feedback of the form: $\widehat{c}(x_i) := \lfloor c(x_i)\rfloor_{\zeta/5}$, where $\lfloor\,\, \rfloor_{\zeta/5}$ denotes rounding to the nearest point in $\In_{\zeta/5}$.  For every target concept $c\in \Cc$ and distribution $D$ on the input space $\mathcal{X}$, define the following subset of $\mathcal{F}$:
\begin{equation}
    G^{\alpha}_{D,\zeta} = \{h \in \mathcal{F}: \Loss_D(h,c,\zeta) \leq \alpha\},
\end{equation}
where $\Loss_D(h,c,\zeta) := \Pr_{x \sim D} \big[|h(x)-c(x)| > \zeta\big]$, so $G^{\alpha}_{D,\zeta}$ may be interpreted as a set of probably-$\zeta$-consistent hypotheses in $\mathcal{F}$. In~\cite{beimel2013}, they show that for every distribution $D$, there exists another distribution~$\tilde{D}$ on the input space, defined as 
\begin{align}\label{eq:tD}
\tilde{D}(x)=\left\{\begin{array}{ll} 1- 4\alpha+4 \alpha \cdot D(x), & x = 0 \\ 4 \alpha \cdot D(x), & x \neq 0\end{array}\right\}
\end{align}
(where $0$ is some arbitrary point in the domain) which has the property
\begin{equation}\label{eq:stability}
    \Pr_{S\sim\tilde{D}^m, \A}\left[\A(S) \in G_{D,\zeta}^{1/4}\right] \geq \frac{3}{4}
\end{equation}
where $\calA(S)$ means $\calA$ is fed with the sample $S$. The property in Eq.~\eqref{eq:stability} follows from the fact that $\Pr_{\tilde{D}}[x] \geq 4\alpha \cdot \Pr_{D}[x] \, \, \forall x \in \mathcal{X}$ by Eq.~\eqref{eq:tD} which implies $G^{\alpha}_{\tilde{D},\zeta} \subseteq G^{1/4}_{D,\zeta}$, and the assumption that $\mathcal{A}$ is $(\zeta,\alpha)$-PAC which can be re-written as
$\Pr_{\tilde{D},\mathcal{A}} [\mathcal{A}(S) \in G^{\alpha}_{\tilde{D},\zeta}] > 3/4$.

Let us now call a sample $S$ `good' if $\vec{x}$ has at least $(1-8\alpha)m$ occurrences of $0$.  Eq.~\eqref{eq:stability} may be rewritten as
\begin{align}\label{ref:rewrite}
    &\Pr_{S\sim\tilde{D}, \A}\left[\A(S) \in G_{D,\zeta}^{1/4}\right] \\
&    =\Pr_{S\sim\tilde{D}, \A} \left[\A(S) \in G_{D,\zeta}^{1/4} \wedge S \text{ is good}\right] + \Pr_{S\sim\tilde{D}, \A} \left[\A(S) \in G_{D,\zeta}^{1/4} \wedge S \,\,\text{is not good}\right] \geq \frac{3}{4}
\end{align}
Letting the random variable $X_S$ denote the number of occurrences of $0$ in $S$, Eq.~\eqref{eq:tD} shows that $\mathbb{E}[X_S] \geq (1-4\alpha)m$. With this we upper bound the term $\Pr_{S\sim\tilde{D}, \A} \left[\A(S) \in G_{D,\zeta}^{1/4} \wedge S \,\,\text{is not good}\right]$~by 
\begin{align}
    \Pr_{S\sim\tilde{D}, \A} \left[S \,\,\text{is not good}\right] &= \Pr_{S\sim\tilde{D}, \A} \left[X_S< (1-8\alpha)m \right]\\
    &= \Pr_{S\sim\tilde{D}, \A}[X_S \leq (1-\delta)(1-4\alpha)m]\leq e^{-\delta^2 (1-4\alpha)m/2} = e^{-2\alpha^2 m/(1-4\alpha)},
\end{align}
where the first inequality used $\delta = \frac{4\alpha}{1-4\alpha}$ and the second inequality follows from a Chernoff bound with $\mathbb{E}[X_S]$ replaced with the upper bound $(1-4\alpha)m$ on its expectation.

Therefore, one can bound the first term on the right hand side of Eq.~\eqref{ref:rewrite} by
\begin{equation}\label{eq:51}
    \Pr_{S\sim\tilde{D}, \A} \left[\A(S) \in G_{D,\zeta}^{1/4} \wedge S \text{ is good}\right] \geq \frac{3}{4} - e^{-2\alpha^2 m/(1-4\alpha)} \geq \frac{1}{4}.
\end{equation}
Eq.~\eqref{eq:51} implies that there exists {\em some} sample, $S_{\rm good}$ such that
\begin{equation}\label{eq:52}
    \Pr_{\A} \left[\A(S_{\rm good}) \in G_{D,\zeta}^{1/4}\right] \geq \frac{1}{4}.
\end{equation}
Without loss of generality we may write down $S_{\rm good}$ as
\begin{equation}
    S_{\rm good} := (\underbrace{(0,\lfloor c(0) \rfloor_{\zeta/5}),\ldots (0,\lfloor c(0) \rfloor_{\zeta/5})}_{\text{$k$ examples}}, (x_{k+1},\lfloor c(x_{k+1})\rfloor_{\zeta/5})\ldots (x_{m},\lfloor c(x_m)\rfloor_{\zeta/5}))
\end{equation}
for some $k\geq (1-8\alpha)m$. Consider an alternative sample, $S_{\rm alt}$, which takes the form $$
S_{\rm alt} = (\underbrace{(0,\lfloor c(0) \rfloor_{\zeta/5}),\ldots, (0,\lfloor c(0) \rfloor_{\zeta/5})}_{\text{$m$ examples}}).
$$ 
$S_{\rm alt}$ differs from $S_{\rm good}$ in exactly $m-k < 8\alpha m$ examples, and so by the $\eps$-$\DP$ property of $\calA$, we~have 
 \begin{equation}\label{eq:54}
     \Pr_{\calA} [\calA(S_{\rm alt}) \in G^{1/4}_{D,\zeta}] \geq \exp(-8\alpha\eps m) \Pr_{\calA} [\calA(S_{\rm good}) \in G^{1/4}_{D,\zeta}] \geq \frac{1}{4} \exp(-8\alpha\eps m).
 \end{equation}
For the remainder of this proof, we will use Eq.~\eqref{eq:54} to construct the pair $(\mathscr{H},\mathcal{P})$. 
Define
$$ 
S_z = (\underbrace{(0, z),\ldots, (0, z)}_{\text{$m$ examples}}).
$$
Now, for each $z \in \In_{\zeta/5}$, run $\calA(S_{z})$ repeatedly $4 \ln (4) e^{8 \alpha \varepsilon m}$ times. Store all the outputs in set $\mathcal{H}$, which has size $|\mathcal{H}| = 5/\zeta \cdot 4 \ln (4) e^{8 \alpha \varepsilon m}$. It is clear that for $z = \lfloor c(0) \rfloor_{\zeta/5}$, $S_{z} = S_{\rm alt}$, and Eq.~\eqref{eq:54} therefore gives us guarantees on the output of $\calA(S_{z})$. We may conclude from Eq.~\eqref{eq:54} that for set $\mathcal{H}$ generated in the above fashion,
\begin{equation}\label{eq:prdim}
    \Pr[\mathcal{H} \cap G^{1/4}_{D,\zeta} = \varnothing] \leq \left(1-\frac{1}{4} e^{-8 \alpha \eps m}\right)^{4 \ln (4) e^{8 \alpha \eps m}} \leq \frac{1}{4}.
    \end{equation}
Rearranging gives $m=\frac{1}{8 \alpha \eps}\left( \PRD_{\zeta,1/4,1/4}(\Cc) - \ln(5/\zeta\cdot 4\ln 4) \right) $.

We may therefore define $\mathscr{H}:= \left\{\mathcal{G} \subseteq \mathcal{F}:|\mathcal{G}| \leq 5/\zeta \cdot 4 \ln (4) e^{8 \alpha \varepsilon m}\right\}$ (note that $\mathcal{H}\in \mathscr{H}$) and further define $\mathcal{P}$ to be the distribution that puts all probability mass on $\mathcal{H}$. Comparing Eq.~\eqref{eq:prdim} with the definition of $\PRD$, Definition~\ref{def:PRD}, observe that $(\mathscr{H},\mathcal{P})$ make up a $(\zeta, 1/4,1/ 4)$ -probabilistic representation for the class $\Cc$. Hence $\textsf{PRDim}_{\zeta,1/4, 1/4}\leq \ln ( 5/\zeta \cdot 4 \ln (4)) + 8 \alpha \varepsilon m$.
\end{proof}

The following lemma is an immediate corollary of~\cite[Theorem 3.1]{feldman2014sample} who proved it for Boolean functions and the exact same proof carries over for our definition of $\PRD$ and randomized one-way communication model in the real-valued setting.

\begin{lemma}[$\PRD$ $\asymp$ Randomized Communication Complexity for real-valued functions]
\label{lem:PRdimtoR}
Let $\Cc$ be a concept class of real-valued functions. The following relations hold:
$$
 \PRD_{\zeta,\eps,\delta}(\Cc) \leq R^{\rightarrow, pub}_{\zeta,\eps\delta}(\Cc),\qquad
  R^{\rightarrow, pub}_{\zeta,\eps+\delta-\eps\delta}(\Cc) \leq \PRD_{\zeta,\eps,\delta}(\Cc).
$$
\end{lemma}

\subsection{One-way communication is characterized by $\sfat(\cdot)$}
\label{sec:onewayandsfat}
We next prove that for every real-valued concept class $\Cc\subseteq \{f:\X\rightarrow [0,1]\}$, the sequential fat-shattering dimension lower bounds the randomized communication complexity of $\Cc$. Namely, we prove the following lemma:
\begin{lemma}\label{lem:R(C)tosfat(C)}
Let $\Cc\subseteq \{f:\X\rightarrow [0,1]\}$ be a concept class. Then
%\begin{equation}
 $   R_{\zeta,\eps}^{\rightarrow}(\Cc) \geq (1-H(\eps))\cdot \sfat_\zeta(\Cc)$.
%\end{equation}
\end{lemma}
With this lemma, we complete our chain of implications, and obtain the conclusion of this section, that the sample complexity of pure $\DP$ $\PAC$ learning upper-bounds the $\sfat(\cdot)$ dimension. 
We remark that the statement above is the real-valued version of the relationship exhibited in~\cite{feldman2014sample}, wherein the Littlestone dimension (Boolean analog of $\sfat(\cdot)$) lower-bounds the randomized communication complexity of Boolean function classes. The proof of Lemma~\ref{lem:R(C)tosfat(C)} proceeds in two steps. First, we define the communication problem $\textsf{AugIndex}_d$ and show that $R_{\zeta,\eps}^{\rightarrow}(\Cc) \geq R_{\eps}^{\rightarrow}(\textsf{AugIndex}_d)$ for $d$ the $\sfat$ dimension of $\Cc$. (We refer the reader to Section~\ref{subsubsec:comm} for the definitions of the quantities $R_{\zeta,\eps}^{\rightarrow}(\cdot)$ and $R_{\eps}^{\rightarrow} (\cdot)$ which pertain respectively to real- and Boolean-function communication complexity.) Next, we use the known relation $R_{\eps}^{\rightarrow}(\textsf{AugIndex}_d) > (1-H(\eps))d$ where $H:[0,1] \rightarrow [0,1]$ is the binary entropy function $H(x):= - x \log x - (1-x) \log (1-x)$. 

To do the first of the two steps, we will relate the one-way classical communication complexities of two communication tasks. The first is the task $\textsf{AugIndex}_d$ for $d\in\mathbb{Z}_+$ which is defined as follows: Alice gets string $x\in \01^d$, while Bob gets $x_{[i-1]}$ for some $i\in[d]$, which is the length-$(i-1)$ prefix of $x$. The task is for Bob to output the bit $x_i$ and we say that $\textsf{AugIndex}_d(x,i) = x_i$. The second is the task $\textsf{Eval}_{\Cc}$, defined in Section~\ref{subsubsec:comm}, for some real-valued function class $\Cc\subseteq \{f:\X\rightarrow [0,1]\}$. We repeat the definition for convenience: Alice is given a function $f\in \Cc$ and Bob a $z\in \X$ and Bob's goal is to approximately compute $f(z)$, i.e., Bob has to compute $b\in [0,1]$ satisfying
\begin{equation}\label{eq:b}
    \Pr_{}\big[|b-f(z)| \leq \zeta\big] \geq 1-\eps,
\end{equation}
where the probability is taken over the local randomness of Alice and Bob respectively. We denote the one-way randomized communication complexity of $\textsf{Eval}_{\Cc}$ as $R_{\zeta,\eps}^{\rightarrow}(\Cc)$ for short. 

\begin{lemma}\label{lem:Ctoaugindex}
If $\Cc\subseteq \{f:\X\rightarrow [0,1]\}$ satisfies $\sfat_\zeta(\Cc)=d$, then $R_{\zeta,\eps}^{\rightarrow}(\Cc) \geq R_{\eps}^{\rightarrow}(\textsf{AugIndex}_d)$.
% \begin{equation}
    
% \end{equation}
\end{lemma}

\begin{proof}
The idea of the proof is to show that a a one-way communication protocol for $\textsf{Eval}_{\Cc}$ can also be used to compute $\textsf{AugIndex}_d$ for $d =\sfat_{\zeta}(\Cc)$. The protocol for $\textsf{AugIndex}_d$ is as follows:
\begin{enumerate}
    \item Alice and Bob agree on the $\zeta$-fat-shattering tree for the concept class $\Cc$ ahead of time.
    \item Upon being given an instance of the $\textsf{AugIndex}_d$ problem, Alice (who has the $d$-bit string $x$) identifies some function in $\Cc$ as follows: she follows the $\zeta$-fat-shattering tree down the path of left-right turns defined by string $x$. This takes her to a leaf $\ell$ which is associated with some unique function $c_{\rm Alice} \in \Cc$.
    Bob (who has the $(i-1)$-bit string $x_{[i-1]}$) identifies some $z_{\rm Bob} \in \mathcal{X}, \,a_{\rm Bob} \in [0,1]$ as follows: he follows the $\zeta$-fat-shattering tree down the path of left-right turns defined by $x_{[i-1]}$. This takes him to some node $w$ at level $i-1$ and Bob sets $z_{\rm Bob},\,a_{\rm Bob}$ to be the domain point and threshold associated with that node. 
    \item Alice and Bob use their protocol $\pi$ for ${\rm Eval}_{\Cc}$ on the inputs $c_{\rm Alice},z_{\rm Bob} $, and following this protocol allows Bob to compute a $b$ that satisfies 
    \begin{equation}\label{eq:b_specific}
    \Pr_{}\big[|b- c_{\rm Alice}(z_{\rm Bob})| \leq \zeta\big] \geq 1-\eps.
    \end{equation}
    \item If $b > a_{\rm Bob}$, Bob outputs 1; else output 0. 
\end{enumerate}
We now prove the correctness of this protocol. Eq.~\eqref{eq:b_specific} states that with probability $1-\eps$, $b$ is a $\zeta$-approximation of $c_{\rm Alice}(z_{\rm Bob})$. Condition on this. In parallel, observe that the Alice's leaf $\ell$ associated with the function $c_{\rm Alice}$ is a descendent of Bob's node $w$ associated with the values $(z_{\rm Bob}, a_{\rm Bob})$, therefore one of the following two statements must be true by definition of $\zeta$-fat-shattering tree and by the procedure outlined in Step 2:
\begin{itemize}
\item $\ell$ is in the right subtree of $w$ i.e., $c_{\rm Alice}(z_{\rm Bob}) > a_{\rm Bob} + \zeta$, and $x_i=1$. By Eq.~\eqref{eq:b_specific}, this implies $b> a_{\rm Bob}$. By Step 4, Bob outputs 1, which is also the value  of $x_i = \textsf{AugIndex}_d(x,i)$.
\item  $\ell$ is in the left subtree of $w$ i.e., $c_{\rm Alice}(z_{\rm Bob}) < a_{\rm Bob} - \zeta$, and $x_i=0$. By Eq.~\eqref{eq:b_specific}, this implies $b < a_{\rm Bob}$. By Step 4, Bob outputs 0, which is also the value  of $x_i = \textsf{AugIndex}_d(x,i)$.
\end{itemize}
This means that the output of Bob in Step 4, $\tilde{b}$, satisfies 
\begin{equation}\label{eq:tb}
    \Pr_{}[\tilde{b}= \textsf{AugIndex}_d(x,i)] \geq 1-\eps,
\end{equation}
where again the probability is taken over the randomness of Alice and Bob. Hence, the protocol above is a valid protocol for computing $\textsf{AugIndex}_d$.
\end{proof}

Finally we can prove the lemma stated at the beginning of the section. 
\begin{proof}[Proof of Lemma~\ref{lem:R(C)tosfat(C)}]
Follows from Lemma~\ref{lem:Ctoaugindex} combined with the inequality $R_{\eps}^{\rightarrow}(\textsf{AugIndex}_d) \geq (1-H(\eps))d$ which was proven in~\cite{feldman2014sample}. 
\end{proof}

In fact, below we strengthen the above into a bound on the one-way {\em quantum} communication complexity of computing real-valued concept classes. 
\begin{corollary}\label{corr:Q(C)tosfat(C)}
Let $\Cc\subseteq \{f:\X\rightarrow [0,1]\}$ be a concept class. Then
%\begin{equation}
 $   Q_{\zeta,\eps}^{\rightarrow}(\Cc) \geq (1-H(\eps))\cdot \sfat_\zeta(\Cc)$.
\end{corollary}
\begin{proof}[Proof of Corollary~\ref{corr:Q(C)tosfat(C)}]
In the proof of Lemma~\ref{lem:Ctoaugindex}, simply replace the classical one-way randomized protocol to compute ${\rm Eval}_C$ with the quantum one-way randomized protocol. This gives that $Q_{\zeta,\eps}^{\rightarrow}(C) \geq Q_{\eps}^{\rightarrow}(\textsf{AugIndex}_d).$ Next, \cite[Theorem~2.3]{nayak} provides a bound for the complexity of quantum {\em serial encoding} that amounts to the statement 
$Q_{\eps}^{\rightarrow}(\textsf{AugIndex}_d) \geq (1-H(1-\varepsilon))d.$ Combining the two yields the claim.  
\end{proof}
\noindent We remark that a similar corollary for \emph{Boolean} concept classes was proven earlier by  Zhang~\cite{Zhang2011OnTP} (where the RHS of Corollary~\ref{corr:Q(C)tosfat(C)} is replaced by Littlestone dimension). Our proof technique easily generalizes to the Boolean setting and significantly simplifies his proof~\cite[Appendix A]{Zhang2011OnTP}.

\section{Applications of our results}
\label{sec:app}
We now present a few applications of the results we established in the previous sections. For the rest of this section, let $\U$ be a class of quantum states on $n$ qubits, and let $\U_n$ refer to the the set of {\em all} quantum states on $n$ qubits. So far, we have shown that the complexity of learning the quantum states from the class $\U$, in two models of learning (pure $\DP$ $\PAC$ and online learning in the mistake bound model), depends on the sequential fat shattering dimension of the real-valued function class $\Cc_{\U}$ associated with $\U$: here $\Cc_{\U} := \{f_{\rho}: \mathcal{X} \rightarrow[0,1]\}_{\rho\in {\U}}$, where $\mathcal{X}$ is the set of all possible two-outcome measurements, and $f_\rho$ is given by $f_{\rho}(E) = \Tr(E \rho)$ for every $E\in \X$. 

In the online learning work of Aaronson et al.~\cite{aaronson2018online} they consider the setting where ${\U}$ is the set of all $n$-qubit states $\U_n$. Let us denote the corresponding function class as $\Cc_n$. In this case,~\cite{aaronson2018online} showed that $\sfat_{\eps}(\Cc_{n}) \leq O(n/\eps^2)$, thus effectively upper-bounding the $\sfat(\cdot)$ dimension of the class of all $n$-qubit quantum states by $n$. This section asks what happens when we allow ${\U} \subseteq {\U}_n$ -- for instance, when ${\U}$ is a special class of states that may be of particular interest or more experimentally feasible to prepare. Are there any meaningful such classes for which we can improve this bound? We first answer this affirmatively for a few classes of quantum states
and finally improve the sample complexity of gentle shadow tomography for these classes of states. 

\subsection{Holevo information and sequential fat shattering dimension}\label{subsec:holevo}
In this section we provide an upper bound on $\sfat(\Cc_{\U})$ in terms of the Holevo information of an ensemble defined on the class of states ${\U}$.
 Using this new upper bound leads to improved upper bounds on $\sfat(\cdot)$ for many classes of quantum states ${\U}$, and hence improved upper bounds on the sample complexity of learning ${\U}$. Previously for $\U=\U_n$, Aaronson~\cite{aaronson2007learnability,aaronson2018online} observed that one could use arguments from quantum random access code by Nayak~\cite{nayak} to obtain a \emph{combinatorial} upper bound on learning.  In this section we show that a better upper bound can be achieved by maximizing the Holevo information, $\chi(\{p_{i}, \rho_{i}\}_{\rho_i\in \U})$ (over all possible distributions~$\vec{p}$ on $\U$), where Holevo information is defined as
\begin{equation}
\chi\left(\left\{p_{i}, \rho_{i}\right\}_{\rho_i \in \U}\right) =S_{}(\bar{\rho})-\sum_{i: \rho_i \in \U} p_{i} S_{}\left(\rho_{i}\right), \quad \bar{\rho}=\sum_{i: \rho_i \in \U} p_{i} \rho_{i}, 
\end{equation}
where $\vec{p}$ is a distribution and $S$ is the von Neumann entropy $S(\rho) := -\Tr[\rho \log \rho]$.

\subsubsection{Quantum Random Access Codes}
 We first define random access codes and serial random access codes over the set $\U$, modifying the definition in~\cite{nayak} so that $\U$ -- the set of states from which the code states may be chosen -- is part of the definition of these codes. 

\begin{definition}[Random access codes and serial random access codes]\label{def:RAC}
Let $\U$ be a class of quantum states over $n$ qubits. A $(k,n,p,\U)$-random access code $(\rac)$ consists of a set of $2^k$ code states $\{\rho_s\}_{s \in \{0,1\}^k}\subseteq \U$ such that, for every $i \in [k]$ and  $s \in\{0,1\}^{k}$, there exists a $2$-outcome measurement~$\mathcal{O}_i$ such that
\begin{equation}\label{eq:rac}
\operatorname{Pr}\left[\mathcal{O}_{i}(\rho_s) =s_{i}\right] \geq p.
\end{equation}
A $(k,n,p,\U)$-{\em serial} random access code $(\srac)$ consists of $2^k$ code states $\{\rho_s\}_{s \in \{0,1\}^k}\subseteq \U$ such that, for every $i \in [k],$ and for all $s \in\{0,1\}^{k}$, there exists a measurement with outcome $0$ or $1$, {\em possibly depending on the last $k-i$ bits $x_{i+1},\ldots ,x_k$}, such that Eq.~\eqref{eq:rac} holds.
\end{definition}

In words, a $\rac$ over $\U$ is a way of encoding $k$ classical bits into $n$-qubit states from $\U$, such that for every $ i \in [k]$  and $x \in\{0,1\}^{k}$, the probability of `recovering' the bit $x_i$ by performing the $2$-outcome measurement $\mathcal{O}_i$ on $\rho_x$ is at least $p$. A serial $\rac$ (denoted $\srac$) is defined similarly except that one is allowed to use information from decoding the {\em subsequent bits} to decode $x_i$. Nayak~\cite{nayak} showed the following  relation between the number of encodable classical bits and the number of qubits in the code states  
\begin{equation}\label{eq:nayak}
    \text{Every $(k,n,p,{\U}_n)$-$\rac$ or $(k,n,p,{\U}_n)$-$\srac$ satisfies $n \geq (1-H(p))k.$}
\end{equation}
Here, $H(\cdot)$ is the binary entropy function, and note that the statement applies to code states drawn from the {\em entire} class of $n$-qubit states.

 Aaronson et al.~\cite{aaronson2018online} in a recent work showed the surprising connection that a $p$-sequential fat-shattering tree for ${\U}$ of depth $k$ can be used to construct a $(k,n,p,{\U})$-$\srac$.\footnote{We remark that such a connection between $\rac$ and learnability was established in an earlier work by Aaronson~\cite{aaronson2007learnability} to understand $\PAC$ learnability of quantum states.} As a corollary of this observation, we have
\begin{equation}\label{eq:aaronson}
\sfat_{p}(\Cc_{\U}) \leq \max \{k: \text{ there exists } (k,n,p,{\U})-\srac\}.
\end{equation}
Combining Eq.~\eqref{eq:nayak},~\eqref{eq:aaronson} yields $\sfat_{p}(\Cc_{\U}) \leq n/(1-H(p))$. In this section, we consider the scenario where $\U \subseteq \U_n$ and show that this bound can be improved to the following.
\begin{theorem}[Bounding $\sfat(\cdot)$ by the Holevo information\label{thm:bettersfat}]
Let $p\in [0,1]$ and $\U$ be some class of quantum states over $n$ qubits. Then  
$$
\sfat_{p}(\Cc_{\U}) \leq \frac{1}{1-H(p)} \max \Big\{\chi\big(\{(q_i, \sigma_i)\}_{\sigma_i\in \U}\big): \sum_i q_i = 1\Big\}.
$$
\end{theorem}

To do so, we tighten the argument of Nayak~\cite{nayak} which was originally derived for ${\U} = {\U}_n$. To prove our result, we make use of the following lemma.
\begin{lemma}[\cite{nayak}]
\label{lem:mixture}
Let $\sigma_{0},\sigma_{1}$ be density matrices and  $\sigma=\frac{1}{2}\left(\sigma_{0}+\sigma_{1}\right)$. If $\mathcal{O}$ is a measurement with $\01$-outcome such that making the measurement on $\sigma_{b}$ yields the bit $b$ with probability $p$,~then
$$
S(\sigma) \geq \frac{1}{2}\left[S\left(\sigma_{0}\right)+S\left(\sigma_{1}\right)\right]+(1-H(p)).
$$
\end{lemma}
We now state and prove our main lemma.
\begin{lemma}\label{lem:RAC}
Let $\U$ be some class of quantum states over $n$ qubits. Every $(k,n,p,\U)$-$\rac$ or $(k,n,p,\U)$-$\srac$ satisfies
\begin{equation}\label{eq:thm_RAC}
(1-H(p))k \leq \max \Big\{\chi\big(\{(q_i, \sigma_i)\}_{\sigma_i\in \U}\big): \sum_i q_i = 1\Big\},
\end{equation}
where $H(\cdot)$ is the binary entropy function and $\chi$ is the Holevo information $\chi\big(\{(q_i, \sigma_i)\}_{\sigma_i\in \U}\big)= S_{}(\sum_i p_i \sigma_i ) - \sum_i p_i S_{}(\sigma_i)$ and $S_{}(\cdot)$ is the von Neumann entropy function.
\end{lemma}

\begin{proof}
Using Definition~\ref{def:RAC}, a $(k,n,p,\U)$-$\rac$ consists of a set of code states $\{\rho_x\}_{x \in \{0,1\}^k}\subseteq \U$ and measurements $\{\mathcal{O}_i\}_{i\in [k]}$ satisfying $\operatorname{Pr}\left[\mathcal{O}_{i}(\rho_x) = x_{i}\right] \geq p.$ Proceeding as in~\cite{nayak}, we first define the following states which are derived from the code states:
For every  $0 \leq \ell \leq k$ and $y \in\{0,1\}^{\ell}$,~let
$$
\sigma_{y}=\frac{1}{2^{k-\ell}} \sum_{z \in\{0,1\}^{k-\ell}} \rho_{z y}.
$$
In words, for a $\ell$-bit string $y$, let $\sigma_y$ be a uniform superposition over all $2^{n-\ell}$ code states with the suffix $y$. Let $\psi  = \frac{1}{2^n} \sum_{z \in\{0,1\}^{n}} \rho_z$ be the uniform superposition over {\em all} code states. Then we have
\begin{align}
\label{eq:rac_3}
S(\psi)\geq \frac{1}{2^k} \sum_{z \in\{0,1\}^{k}} S(\rho_z) + k (1-H(p)).
\end{align}
To see this, first one can use Lemma~\ref{lem:mixture} to show $S(\psi) \geq \frac{1}{2} \left(S(\sigma_0) + S(\sigma_1) \right) + 1-H(p)$ and recursively applying this lemma to each of the $S(\cdot)$ quantities, we get the equation above (observe that each application of the lemma is justified because for every $y \in \{0,1\}^{\ell}$, we may write $\sigma_y = \frac{1}{2} \left(S(\sigma_{0y}) + S(\sigma_{1y}) \right)$; and  by assumption of a $(k,n,p,\U)$-$\rac$, $\mathcal{O}_{\ell+1}$ can distinguish $\sigma_{0y}, \sigma_{1y}$ with success probability $p$ and thus is a measurement that meets the conditions of Lemma~\ref{lem:mixture}.) 
 Using Eq.~\eqref{eq:rac_3} it now follows that
\begin{align}\label{eq:rac_4}
k (1-H(p)) &\leq S_{}(\psi) - \frac{1}{2^k} \sum_{z \in\{0,1\}^{k}} S_{}(\rho_z)  = \chi\Big(\big\{ \frac{1}{2^k},\rho_x\big\}_{x\in \{0,1\}^k}\Big) \leq \max_{T\subseteq \U} \chi\Big( \Big\{\frac{1}{|T|}, \sigma_i\Big\}_{\sigma_i\in T} \Big).
\end{align}
where the last inequality follows because the uniform ensemble of code states  $\big\{\frac{1}{2^k},\rho_x\big\}_{x\in \{0,1\}^k}$ is precisely of the form $\{p_i, \sigma_i\big \}_{\sigma_i\in \U}$ with zero weight on non-code states in $\U$. In Eq.~\eqref{eq:thm_RAC}, to get a simpler-looking bound, we further relax this inequality by taking the optimization over arbitrary probability distributions on the code states, not just the ones that are uniform on a subset.  Eq.~\eqref{eq:thm_RAC} also holds for $\srac$ by noting that the argument above doesn't change by allowing $\mathcal{O}_i$ to depend on bits $x_{i+1},\ldots ,x_k$. 
\end{proof}
\noindent The proof of Theorem~\ref{thm:bettersfat} follows immediately from combining Lemma~\ref{lem:RAC} and Observation~\eqref{eq:aaronson}. 

An interesting consequence of our result is the following. As far as we are aware,  there is no way of computing $\sfat(\cdot)$ directly, but there exist algorithms to compute our bound in Theorem~\ref{thm:bettersfat}. For a set $\U$ of states, performing the maximization $\max \{\chi\big(\{(q_i,\sigma_i)\}\big):  \sum_i q_i = 1\}$ is a convex optimization problem which can be solved using the Blahut-Arimoto algorithm\cite{blahut}. However, for certain special classes of states, one can present simple bounds on the maximal Holevo information  which we present next.

\subsubsection{Classes of states with bounded $\sfat(\cdot)$ dimension}\label{sec:smallsfat}
A natural question is, how does the new upper bound on $\sfat(\U)$ in Theorem~\ref{thm:bettersfat} compare to the previous upper bound $\sfat(\U_n) < n/\eps^2$ given in~\cite{aaronson2018online}. Observe that that the $\eps$ dependence comes about from a Taylor expansion of $1-H((1-\eps)/2)$ and our new bounds do not change this dependence, hence for the remainder of this section we set $\eps=1$ for simplicity. We now mention a few classes of states for which our new bound improves the $n$ dependence of the previous bound.

\begin{itemize}
\item Suppose our quantum states are ``$k$-juntas", i.e.,      each  $n$-qubit quantum state lives in the same \emph{unknown} $k$-dimensional subspace of the $2^n$-dimensional Hilbert space. Then clearly, the right-hand-side of Eq.~\eqref{eq:thm_RAC} is upper-bounded by $\log k < n$. 
In particular for $n$-juntas the $\sfat(\cdot)$ dimension is $O(\log n)$, hence the sample complexity of learning scales as $O(\log n)$ which is exponentially better than the prior upper bounds of~$n$. 
\item $\U$ consists of a small set of states with small pairwise trace distance; in~\cite{Audenaert_2014} and~\cite{shirokov} they showed that 
\begin{equation}\label{eq:highd}
    \chi(\{p_i,\rho_i\}) \leq v_m \log |\U| 
\end{equation}
where $v_{\mathrm{m}}=\frac{1}{2} \sup _{i, j}\left\|\rho_{i}-\rho_{j}\right\|_{1}$ is the maximal trace norm distance between the states in the class $\U$. This bound could be significantly better than the trivial $\log |\U|$ if $v_m$ is sufficiently~small. 
\item Let $\U = \mathcal{N}(\U_n)$ be the set of all $n$-qubit states obtained after passing the states in $\U_n$ through the channel $\mathcal{N}$. That is, we would like to learn some arbitrary $n$-qubit state that has been passed through an \emph{unknown} quantum channel $\mathcal{N}$. This is the case in many experimentally-relevant settings and is in fact one way to  understand the effect of experimental noise (which can be modelled by a quantum channel during state preparation).   The Holevo information of the quantum channel $\mathcal{N}$ is the following quantity
\begin{equation}\label{eq:holevo}
\chi(\mathcal{N}) := \max_{\vec{p},\rho_i} S_{}\Big(\sum_i p_i \mathcal{N}(\rho_i)\Big) - \sum_i p_i S_{}(\mathcal{N}(\rho_i)),
\end{equation}
where the maximization is over  (arbitrary-sized) ensembles $\{(p_i,\rho_i)\}$. Observe that using Eq.~\eqref{eq:thm_RAC} one can upper bound $\sfat(\cdot)$ dimension of the set $\U = \mathcal{N}(\U_n)$ in terms of $\chi(\mathcal{N})$. 
A centerpiece of quantum Shannon theory is the Holevo-Schumacher-Westmoreland (HSW) theorem~\cite{HSW}, which states that (see for example~\cite{wilde} for a pedagogical proof) $
\chi(\mathcal{N}) \leq~C(\mathcal{N})
$
where $C(\mathcal{N})$ is the classical capacity of the channel. Putting these two bounds together gives
\begin{equation}\label{eq:capacitybound}
\sfat(\mathcal{N}(\U_n)) \leq C(\mathcal{N}).
\end{equation}
Now, using the connection above one can upper bound $\sfat(\cdot)$ of noisy quantum states using results developed in quantum Shannon theory to bound the classical channel capacity. For a depolarizing channel acting on $d$-dimensional states with parameter $\lambda$ for instance (a common noise model), one can upper bound $C(\mathcal{N})$ in  Eq.~\eqref{eq:capacitybound} by a result of King~\cite{king} as follows
\begin{equation}
    \log d - S_{\min }\left(\Delta_{\lambda}\right)
\end{equation}
where $S_{\min }\left(\Delta_{\lambda}\right)=-\left(\lambda+\frac{1-\lambda}{d}\right) \log \left(\lambda+\frac{1-\lambda}{d}\right)-(d-1)\left(\frac{1-\lambda}{d}\right) \log \left(\frac{1-\lambda}{d}\right)$ and the subtractive quantity in the quantity above makes this bound strictly better than~\cite{aaronson2018online}. Similar upper bounds on channel capacity are also known for Pauli channels~\cite{siudzinska2019regularized} and generalized Pauli channels~\cite{siudzinska2020classical}. 

\item Interestingly, we may now also bound $\sfat(\cdot)$ of the class of quantum Gaussian states. Since these states are infinite-dimensional, the previous bound of~\cite{aaronson2018online} is not useful.
However, our channel capacity upper-bound on $\sfat(\cdot)$ yields a finite bound:
It is known from~\cite{giovanetti2004} that the channel capacity of a pure-loss Bosonic channel  with transmissivity $\eta \in [0,1]$,\footnote{This channel is a simple model for communication over free space or through a fiber optic link, where $\eta$ models how much noise is `mixed' into the states.} when the input Gaussian states have photon number at most $N_p$ (and hence bounded energy, which is physically realistic), is $g(\eta N_p)$ where $g(x) \equiv(x+1) \log _{2}(x+1)-x \log _{2} x$. In particular, the case $\eta=1$ corresponds to zero loss, hence $g(N_p)$ bounds $\sfat(\cdot)$ for the {\em entire} class of Gaussian states with $N_p$ photons. 

Alternatively, one might be interested in states prepared through phase-insensitive Bosonic channels. These model other kinds of noise, such as thermalizing or amplifying processes. A recent work~\cite{giovannetti_2014} allows one to bound the capacities of these channels, and hence the $\sfat(\cdot)$ dimensions of these noisy Gaussian states. 
\end{itemize}

\subsection{Faster online shadow tomography}\label{subsec:shadow}
We now discuss how our results can also improve \emph{shadow tomography}, a learning framework recently introduced by Aaronson~\cite{aaronson:shadow}. This is a variant of quantum state tomography in which the goal is not to learn $\rho$ completely, but to learn its `shadows', i.e.,  the expectation values of $\rho$ on a fixed (known) set of measurements. 

To be precise, let $\U$ be a subset of $n$-qubit states. Given $T$ copies of an unknown state $\rho \in \U$, and a set of known two-outcome measurements $E_{1}, \ldots, E_{m}$. The goal is to learn (with probability at least $2/3$) $\Tr(E_{i} \rho)$ up to  additive error $\varepsilon$ for every $i \in[m]$. A trivial learning algorithm uses $T=O((2^n+m)\cdot \varepsilon^{-2})$ many copies of $\rho$  to solve the task, and surprisingly Aaronson showed how to solve this task using  $T=\poly(n,\log m,\varepsilon^{-1})$ copies of $\rho$, exponentially better than the trivial algorithm. An intriguing open question left open by Aaronson~\cite{aaronson:shadow} and others is, is the $n$ dependence necessary? There have been follow up results by Huang et al.~\cite{huangpreskill} that improved Aaronson's procedure when the goal is to obtain `classical shadows' and more recently B{\u{a}}descu and O'Donnell~\cite{buadescu2020improved} gave a procedure which has the best known dependence on all parameters for standard shadow tomography. 

Subsequently Aaronson and Rothblum~\cite{aaronson2019gentle} considered \emph{gentle} shadow tomography, a (stronger) variant of shadow tomography (we do not define gentleness here and refer the interested reader to~\cite{aaronson2019gentle}).  Here, we show that suppose we were performing gentle shadow tomography with the prior knowledge that the unknown state $\rho$ came from a class of states $\U$, then the $n$-dependence in the sample complexity can be replaced with $\sfat(\Cc_\U)$. As we discussed in the previous section, clearly $\sfat(\Cc_\U)\leq O(n/\varepsilon^2)$, but for many class of states $\sfat(\Cc_\U)$ could be much lesser than $n$, giving us a significant improvement over Aaronson's result.  We first state our main statement.  %
\begin{theorem}[Faster gentle shadow tomography]
\label{thm:onlineimpliesST}
The complexity of gentle shadow tomography on a class of states $\U$ is
\begin{equation}\label{eq:gentle}
O\left(\frac{\sfat_{\eps}(\Cc_{\U})^2 \log^2 m \log(1/\delta)}{\eps^2 \min\{\alpha^2,\eps^2\}}\right).
\end{equation}
where $\alpha, \delta$ are gentleness parameters and the goal is to learn $\Pr[E_i (\rho) \text{ accepts}]$ to within an additive error of $\eps$ for every $i\in [m]$.\footnote{Implicitly in the complexity above we have assumed that the algorithm succeeds with probability at least $2/3$.} Moreover, there exists an explicit algorithm that achieves this.
\end{theorem}

Indeed the parameter $\sfat(\Cc_{\U})$ in this bound means that for the classes of states mentioned in Section~\ref{sec:smallsfat}, the sample complexity of shadow tomography is better than the complexity in~\cite{aaronson:shadow} (in terms of $n$). We now prove Theorem~\ref{thm:onlineimpliesST}. The connection comes from the implication in~\cite{aaronson2019gentle} that under certain conditions, an online learner for quantum states can be used as a black box for what they term `Quantum Private Multiplicative Weights', an algorithm that performs shadow tomography in both an online and a gentle manner.  We now state the precise setting in which this black box online learner must operate. As usual, we are concerned with the function class $\Cc_{\U} := \{f_{\rho}\}_{\rho\in {\U}}$ where the domain $\mathcal{X}$ is the set of all possible two-outcome measurements $E$ on the states in ${\U}$ and the functions in the class are defined as $f_{\rho}(E) = \Tr(E \rho)$ for every $E$. The unknown state $\rho$ defines some target function $c \in \Cc_{\U}$.
\begin{enumerate}
    \item Adversary provides input point in the domain: $x_t \in \mathcal{X}$.
    \item Learner outputs a prediction $\hat{y}_t \in [0,1]$. 
    \item If the learner makes a mistake, i.e.,     if $|\hat{y}_t - c(x_t)| > \eps$, then adversary provides strong feedback $\widehat{c}(x_{t}) \in [0,1]$ where $\widehat{c}(x_t)$ is an $\varepsilon/10$-approximation of $c(x_t)$, i.e.,       $|\widehat{c}(x_t)- c(x_t)| < \eps/10$, and the learner is allowed to update its hypothesis. Else, the adversary does not provide any feedback, and the learner must use the same hypothesis on the next round.
    \item Learner suffers loss $\left|\hat{y}_t - c(x_t)\right|.$
\end{enumerate}

Observe that this is a close variant of our setting in Section~\ref{subsec:learning_models}, the only difference being that the adversary here only gives feedback on rounds where the learner makes a mistake (i.e.,      when the learner's prediction is grossly wrong). This means that the learner updates her hypothesis if and only if it makes a mistake. Given an online learner $\mathcal{A}$ in the above setting that makes at most $\ell$ updates,~\cite[Theorem 38]{aaronson2019gentle} shows that there exists a randomized algorithm $\mathcal{B}$ for shadow tomography using
\begin{equation}\label{eq:qpmw}
n = O\left(\frac{\ell^2 \log^2 m \log(1/\delta)}{\eps^2 \min\{\alpha^2,\eps^2\}}\right).
\end{equation}
many examples of the unknown state $\rho$ where
such that algorithm $\mathcal{B}$'s error is bounded by $\varepsilon$ with probability at least $1-\beta$. Moreover, this algorithm is $(\alpha,\delta)$-gentle. 
We are now equipped with all we need to prove Theorem~\ref{thm:onlineimpliesST}. The proof boils down to the observation that for any concept class~$\Cc$, we can always construct an online learner that is guaranteed to make at most $\sfat(\Cc)$ mistakes, and therefore $\ell =\sfat(\Cc)$ in Eq.~\eqref{eq:qpmw}. The online learner we construct is a variant of the proper version of our $\RSOA$ Algorithm~\ref{algo:RSOA}.

\begin{proof}[Proof of Theorem~\ref{thm:onlineimpliesST}]
The proof follows from the Quantum Private Multiplicative Weights algorithm in~\cite{aaronson2019gentle} and its accompanying Theorem 39, simply by exhibiting an online learner $\mathcal{A}$ for $\U$ in the setting described above, that makes at most $\ell = \sfat_{\eps}(\Cc_S)$ mistakes. In the rest of this proof, we exhibit just such an algorithm, which is a variant of the proper version of $\RSOA$.

The difference between Algorithm~\ref{algo:RSOA_variant} and $\RSOA$ is that in $\RSOA$, the learner is allowed to update the set $V_t$ on all rounds $t\in [T]$, while in Algorithm~\ref{algo:RSOA_variant}, the update happens only on the rounds for which it made a mistake (`mistake rounds'). Because the learner's current hypothesis for the target concept is computed based on the `set of surviving functions' $V_t$, updating $V_t$ amounts to updating the algorithm's hypothesis.
  \begin{algorithm}[H]
		\textbf{Input:} Concept class $\Cc\subseteq \{f:\X\rightarrow [0,1]\}$, target (unknown) concept $c\in \Cc$, and $\varepsilon\in [0,1]$.
		\vspace{5pt}
		\textbf{Initialize}: $V_1 \gets \Cc$
		\vspace{5pt}
		\begin{algorithmic}[1]
\For{$t = 1, \ldots, T$}
\vspace{1mm}
    \State A learner receives $x_t$ and maintains set $V_t$, a set of ``surviving functions". \;
    \State For every super-bin midpoint $r\in \tilde{\In}_{2\eps/5}$
    the learner computes the set of functions $V_t(r,x_t)$. 
    \State A learner finds the super-bin which achieves the maximum $\sfat(\cdot)$ dimension
      $$
     R_{t}(x_t):=\left\{
    \argmax _{r \in \tilde{\In}_{2\eps/5}} \sfat_{2\eps/5}\left(V_{t}(r, x_t)\right)\in \tilde{\In}_{2\eps/5}\right\}
    $$
    \State The learner computes the mean of the set $R_t(x_t)$, i.e.,      let 
    $$
    \hat{y}_t:=\frac{1}{\left|R_{t}(x_t)\right|} \sum_{r \in R_{t}(x_t)} r.
    $$ 
    \State The learner outputs $\hat{y}_t$,  receives feedback $\widehat{c}(x_{t})$ if it has made a mistake, i.e.,      if $|\hat{y}_t - c(x_t)| > \eps$.     
    \State If the learner received feedback, update $V_{t+1} \leftarrow \{g \in V_t \mid g(x_t) \in B_{\eps/5}(\widehat{c}(x_{t}))\}$; else $V_{t+1} \leftarrow~V_{t}$. 
\EndFor
\end{algorithmic}
\textbf{Outputs:} The intermediate predictions $\hat{y}_t$ for $t\in[T]$, and a final prediction function/hypothesis which is given by $f(x):= R_{T+1}(x)$.

\caption{\textsf{Alternative Robust Standard Optimal Algorithm}
}
\label{algo:RSOA_variant}
\end{algorithm}
  We thus aim to show that Algorithm~\ref{algo:RSOA_variant} has no more than $\sfat(\cdot)$ mistake rounds. However, we observe that we may directly import the proof of Theorem~\ref{thm:RSOA} to do so. This is because that proof is independent of what happened on the non-mistake rounds, which are the only rounds that differ between $\RSOA$ and Algorithm~\ref{algo:RSOA_variant}. Rather, it argues that on the rounds on which $\RSOA$ made a mistake, $\sfat(V_t)$ decreases by at least $1$ due to the update on $V_t$, and having initialized $V_1=\Cc$, no more than $\sfat(\Cc)$ updates may happen in total. Exactly the same argument can be used to bound the mistakes of Algorithm~\ref{algo:RSOA_variant}, though note that for the constants to work out, the $\eps$ of $\RSOA$ must be multiplied by 5.
\end{proof}

\bibliographystyle{alpha}
\bibliography{online}
\end{document}